\documentclass[11pt]{article}

\usepackage[top=1in,bottom=1in,left=1in,right=1in]{geometry}
\usepackage{natbib}
\usepackage[unicode=true,bookmarks=true,bookmarksnumbered=false,bookmarksopen=false,breaklinks=false,pdfborder={0 0 1},backref=false,colorlinks=true]{hyperref}
\hypersetup{citecolor=blue}
\usepackage{url}

\usepackage{amsmath}
\usepackage{mathrsfs}
\usepackage{amssymb}
\usepackage{amsthm}
\usepackage[normalem]{ulem}

\newtheorem{theorem}{Theorem}

\newtheorem{assumption}{Assumption}

\usepackage{float}
\usepackage{rotfloat}

\floatstyle{ruled}
\newfloat{algorithm}{tpb}{loa}
\providecommand{\algorithmname}{Algorithm}
\floatname{algorithm}{\protect\algorithmname}

\usepackage{algpseudocode,algorithm}
\algnewcommand\algorithmicinput{\textbf{Input}:}
\algnewcommand\algorithmicoutput{\textbf{Output}:}
\algnewcommand\INPUT{\item[\algorithmicinput]}
\algnewcommand\OUTPUT{\item[\algorithmicoutput]}

\usepackage{graphicx}
\usepackage[caption=false,labelformat=simple]{subfig}

\graphicspath{}

\usepackage{array}
\newcolumntype{L}[1]{>{\raggedright\let\newline\\\arraybackslash\hspace{0pt}}m{#1}}
\newcolumntype{C}[1]{>{\centering\let\newline\\\arraybackslash\hspace{0pt}}m{#1}}
\newcolumntype{R}[1]{>{\raggedleft\let\newline\\\arraybackslash\hspace{0pt}}m{#1}}

\usepackage{rotating}
\usepackage{multirow}

\usepackage{tikz}
\usetikzlibrary{shapes,arrows,backgrounds,calc,positioning,fit,petri,plotmarks}
\usetikzlibrary{arrows}

\usepackage{enumerate}
\usepackage{paralist}

\newcommand*{\affaddr}[1]{#1} 
\newcommand*{\affmark}[1][*]{\textsuperscript{#1}}

\global\long\def\bx{\mathbf{x}}
\global\long\def\bY{\mathbf{Y}}
\global\long\def\by{\mathbf{y}}

\global\long\def\bD{\mathbf{D}}
\global\long\def\bA{\mathbf{A}}

\global\long\def\bC{\mathbf{C}}
\global\long\def\bH{\mathbf{H}}

\global\long\def\bw{\mathbf{w}}
\global\long\def\bU{\mathbf{U}}

\global\long\def\bV{\mathbf{V}}

\global\long\def\bQ{\mathbf{Q}}

\global\long\def\bS{\mathbf{S}}

\global\long\def\bK{\mathbf{K}}

\global\long\def\matzero{\boldsymbol{\mathrm{0}}}
\global\long\def\matI{\boldsymbol{\mathrm{I}}}

\global\long\def\bbeta{\boldsymbol{\beta}}

\global\long\def\bSigma{\boldsymbol{\Sigma}}
\global\long\def\bgamma{\boldsymbol{\gamma}}

\global\long\def\bOmega{\boldsymbol{\Omega}}
\global\long\def\bPi{\boldsymbol{\Pi}}
\global\long\def\bpi{\boldsymbol{\pi}}
\global\long\def\bLambda{\boldsymbol{\Lambda}}

\global\long\def\bxi{\boldsymbol{\xi}}
\global\long\def\bGamma{\boldsymbol{\Gamma}}
\global\long\def\bTheta{\boldsymbol{\Theta}}
\global\long\def\bvartheta{\boldsymbol{\vartheta}}

\global\long\def\diag{\mathrm{diag}}

\usepackage{xr}
\makeatletter
\newcommand*{\addFileDependency}[1]{
  \typeout{(#1)}
  \@addtofilelist{#1}
  \IfFileExists{#1}{}{\typeout{No file #1.}}
}
\makeatother
 
\newcommand*{\myexternaldocument}[1]{%
    \externaldocument{#1}%
    \addFileDependency{#1.tex}%
    \addFileDependency{#1.aux}%
}

\myexternaldocument{}

\title{Multilevel Regression Modeling of Covariance Matrix Outcomes}

\author{%
    Michelle Murphy Green\affmark[1], Xi Luo\affmark[2], Brian S. Caffo\affmark[3], and Yi Zhao\affmark[1*]
    \\
    \affaddr{\affmark[1]Department of Biostatistics and Health Data Science, \\ Indiana University School of Medicine} \\
    \affaddr{\affmark[2]Department of Biostatistics and Data Science,\\  The University of Texas 
Health Science Center at Houston} \\
    \affaddr{\affmark[3]Department of Biostatistics, Johns Hopkins Bloomberg School of Public Health} \\
}

\date{}

\begin{document}

\maketitle

\thispagestyle{empty}

\begin{abstract}
Covariance matrix outcomes arise naturally in neuroimaging experiments to study brain functional connectivity. It is also of interest to understand how brain network organization varies with subject-level covariates. Existing covariance regression methods operate in a single-level framework and do not accommodate the hierarchically nested data structure in which subjects are grouped into clusters, such as age cohorts in lifespan studies. A Multilevel Covariate-Assisted Principal Regression (MCAP) framework is introduced, which identifies, for each cluster, a linear projection such that a generalized linear mixed effects model can be formulated with the covariates. The cluster-specific projections are modeled on the unit sphere via a von Mises-Fisher distribution, enabling principled borrowing of information across clusters. Model parameters are estimated by maximizing a hierarchical likelihood. For inference, a two-stage bootstrap procedure is proposed. Asymptotic properties of the estimators are established. Simulation studies demonstrate that MCAP substantially outperforms single-level competitors in estimating regression coefficients. Applied to the Human Connectome Project Lifespan Study spanning ages from five to ninety, MCAP identifies a dominant spectral brain network capturing age and sex effects on functional connectivity, and reveals findings including the convergence of neural reorganization patterns in late adulthood and the coordinated lifespan modulation of cross-network regions linked to language and executive function.
\end{abstract}



\clearpage
\setcounter{page}{1}

\section{Introduction}
\label{sec:intro}

Understanding how the human brain is organized and interconnected is necessary to understand cognition and behavior. From neuron-to-neuron synapses to neural pathways between regions to whole-brain networks, this complex system undergoes continuous restructuring throughout our lives. One method used to make sense of this system is functional magnetic resonance imaging (fMRI), in which blood flow in the brain is used to identify brain activity. Previous fMRI research has found functional connectivity (FC) to be constantly adapting and changing throughout the lifespan, with a growth curve that peaks during the late third and fourth decades~\citep{sun2025human}. In essence, FC in late degeneration stages may mirror patterns observed in early development. Research into neuroplasticity, or the brain's ability to reorganize and restructure in response to stressors, supports this variation in FC across the lifespan. Specifically, there is a U-shaped pattern in the need for functional brain network change, with greater demand during adolescence and late adulthood and comparatively little FC change during adulthood~\citep{saberi2021requirement}. The increased demand for network restructuring during late adulthood likely stems from the breaking down of previously constructed systems and an attempt to sustain and maintain FC. This is further emphasized by a corresponding bell-shaped curve in grey matter volume across the lifespan~\citep{bethlehem2022brain}.
However, development varies by region with certain regions such as default mode network (DMN) and frontoparietal (FP) regions showing comparatively greater changes through rapid development in early life followed by rapid decline during aging~\citep{supekar2010development, ng2024frontoparietal}. While the understanding of global FC changes across the lifespan has improved greatly in the past decade, the next step is to uncover region-specific changes and their impact on the developing brain. 

Sex differences have also been studied, although the findings are not fully consistent. Some studies have identified higher global FC in males than in females~\citep{zhang2016sex, zhang2024different}. Others have found regional differences in which FC is occasionally higher in females, with the affected regions depending on menopausal state~\citep{kilpatrick2025differential}. Sex can also be predicted from FC, predominantly in the DMN~\citep{zhang2018functional}. Furthermore, sex differences have been reported in the organization of brain networks, with male networks showing stronger intrahemispheric connections and female networks showing stronger interhemispheric connections~\citep{ingalhalikar2014sex}. As such, sex differences in FC and brain networks remain an important topic, and these findings highlight the importance of accounting for sex in research questions and further exploring these differences in future work.

The Human Connectome Project (HCP), launched in 2010, is a large-scale coordinated effort to comprehensively map the structural and functional connectivity of the human brain, aiming to elucidate the organization of neural circuits and their relationship to behavior. Utilizing advanced techniques in neuroimaging and combining three cohorts, HCP Development (HCP-D), HCP Young Adult (HCP-YA), and HCP Aging (HCP-A), it enables the investigation of brain connectivity from more than $3,000$ subjects spanning ages from five to ninety on a scale not previously available~\citep{van2013wu, elam2021human}. With these cross-sectional cohorts, population-level profiles along the lifespan can be postulated by creating age-defined clusters. This hierarchically nested data structure motivates the development of a multilevel modeling framework proposed in this manuscript.

The statistical analysis of FC from fMRI studies centers on the covariance matrix of time courses extracted from brain regions of interest~\cite[ROIs,][]{biswal1995functional}. A covariance matrix outcome, however, presents fundamental modeling challenges: it lives on the cone of positive semidefinite matrices, has $p(p+1)/2$ unique entries that grow quadratically in dimension, and its entries exhibit complex within-matrix dependence. A common simplification is to vectorize the lower-triangular of the covariance matrix and apply standard multivariate analysis. However, this ignores the positive-definiteness constraint and typically requires a large number of parameters relative to sample size, resulting in reduced statistical power. 
An alternative that has gained traction focuses on modeling the full covariance matrix as the outcome while preserving positive definiteness. Such approaches directly model the covariance matrix as a function of covariates through linear-combination models~\citep{anderson1973asymptotically}, matrix-logarithm regression~\citep{chiu1996matrix}, quadratic covariance regression and its extensions~\citep{hoff2012covariance,fox2015bayesian,seiler2017multivariate}, similarity-matrix formulations~\citep{zou2017covariance}, Bayesian hierarchical models for heterogeneous covariance matrices~\citep{acharyya2023bayesian}, common-diagonalization approaches based on eigendecomposition~\citep{flury1984common,boik2002spectral,hoff2009hierarchical,franks2019shared}, and Cholesky-based parameterizations for ordered outcomes~\citep{pourahmadi2007simultaneous}. More recent work has extended these ideas to high-dimensional settings through sparse covariance regression and covariate-dependent Cholesky decompositions, improving scalability and interpretability while continuing to face tradeoffs among flexibility, computational tractability, and guaranteed positive definiteness~\citep{kim2025high,kim2026covariate}. 
Another direction is to seek low-dimensional representations of covariance variation rather than modeling every matrix entry directly. In particular, covariate-assisted principal (CAP) regression identifies projection directions such that the projected variance follows a generalized linear model in the covariates, thereby targeting interpretable covariance components associated with subject-level characteristics~\citep{zhao2021covariate}. This projection-based strategy is computationally attractive and avoids the multiplicity and interpretability difficulties of full-matrix regression. \citet{zhao2024longitudinal} extended this framework to repeated covariance matrices through a mixed-effects model that captures within-subject variation and accommodates time-varying predictors, while assuming the projection direction to be common across subjects and time points. Related multilevel covariance regression and factor-analytic formulations have also been proposed for hierarchical data~\citep{orindi2023combined}, but do not allow the covariance-associated subspace itself to vary across clusters. 
Recent network-based methods in neuroimaging have emphasized that connectivity edges should not be treated as independent and that preserving higher-order network structure can materially improve statistical power and interpretability~\citep{wang2026latent}. 

By defining age clusters, the HCP Lifespan data are framed with a multilevel structure. We therefore propose a hierarchical extension of the CAP approach, named Multilevel Covariate-Assisted Principal Regression (MCAP), that allows the projection direction to vary across clusters. Specifically, cluster-specific projection vectors are modeled as random directions on the unit sphere through a von Mises-Fisher distribution within a hierarchical modeling framework. This formulation preserves the interpretability of CAP through projected covariance components while relaxing the restrictive common-projection assumption in existing longitudinal principal regression literature. The proposed framework accommodates between-cluster heterogeneity in covariance-associated subspaces, offering a more flexible framework for multilevel covariance matrix outcomes. 

The rest of the manuscript is organized as follows. Section~\ref{sec:motivation} introduces the motivating HCP Lifespan Study. The proposed multilevel covariance matrix modeling framework is presented in Section~\ref{sec:model}, together with the estimating algorithm, inference strategy, and asymptotic properties. Section~\ref{sec:sim} illustrates the performance of the proposed framework through simulation studies. Section~\ref{sec:hcp} revisits the HCP study and demonstrates the constructed lifespan profile of brain functional connectivity and corresponding variation in brain subnetworks across age clusters. Finally, Section~\ref{sec:discussion} summarizes the manuscript with discussions.


\section{Motivation: The Human Connectome Project Lifespan Study}
\label{sec:motivation}


Brain reserve and cognitive reserve are known to be complementary mechanisms that assist the brain in resisting damage from aging or disease. Brain reserve reflects the structural neurobiological capacity, such as brain volume and synaptic density. Conversely, cognitive reserve reflects the brain's utilization of brain networks and the functional adaptability that enables individuals to maintain cognitive performance despite pathology~\citep{stern2002cognitive,stern2020whitepaper}. As such, cognitive reserve is a key determinant of resilience to brain aging and neurodegernation and individuals with diminished cognitive reserve are at increased risk for age-related cognitive decline and dementia~\citep{livingston2020dementia}. It is known that cognitive reserve is shaped by dynamic interactions between biological, cognitive, and social factors across the lifespan, including early-life adversity and social determinants of health~\citep{luby2013effects,kaup2014older,saberi2021requirement,sun2025human}. However, its lifespan trajectories and underlying mechanisms remain poorly understood. This is largely due to the lack of harmonized individual-level longitudinal data across the various stages of the life course, and therefore continues to be a fundamental limitation in determining lifespan trajectories of cognitive reserve~\citep{bethlehem2022brain,mousley2025topological}. 

The Human Connectome Project (HCP) Lifespan Studies provide a compelling setting to address this persistent, longitudinal gap in cognitive reserve research. Building on the original HCP Young-Adult (HCP-YA) cohort, which established high-quality protocols for mapping structural and functional brain connectivity, the Lifespan extension broadened this framework to HCP-Development (HCP-D) and HCP-Aging (HCP-A), thereby covering childhood, adolescence, adulthood, and older age with broadly harmonized imaging and preprocessing pipelines~\citep{van2013wu,harms2018extending,bookheimer2019lifespan,somerville2018lifespan}. Across these cohorts, participants undergo multimodal neuroimaging, including structural MRI, diffusion MRI, resting-state fMRI, and task fMRI, together with extensive cognitive, behavioral, and health assessments. These data create a unique opportunity to investigate how cognitive reserve measurements vary with age and other subject-level characteristics across lifespan.

In this study, we focus on the construction of brain functional connectivity profile over the lifespan. 
For each individual, subject-specific covariance or connectivity matrices are derived from resting-state fMRI scans that encode coordinated activity among brain regions. These matrices are likely to vary systematically with developmental stage, aging status, sex, cognition, and related reserve-associated factors. At the same time, the scientific question is not only whether the magnitude of connectivity differs with age, but also whether the brain subnetwork, represented by covariance-associated subspace, itself changes across life stages. This consideration motivates the use of age-defined clusters, such as developmental, young-adult, and aging groups, as a principled way to capture major differences in neurobiological context, reserve formation, and reserve depletion. Such clustering is scientifically meaningful because the mechanisms governing covariance structure are unlikely to be exchangeable across childhood neurodevelopment, mature adult maintenance, and late-life aging, and it is statistically useful because it allows the model to pool information within relatively homogeneous age strata while permitting structured heterogeneity across strata.

Under this consideration, a common linear projection for all individuals, as assumed in \citet{zhao2024longitudinal}, may be overly restrictive. Across age groups, covariance components may vary in orientation reflecting different latent subnetworks involved in reserve accumulation, maintenance, or decline at different stages of life. This motivates a multilevel covariance regression framework in which age-defined clusters serve as higher-level units and the covariance components are allowed to vary across clusters. Such a model separates population-level covariate effects from cluster-specific deviations in covariance geometry, thus better accommodates lifespan heterogeneity in brain functional organization. The HCP Lifespan Study offers an ideal opportunity for the proposed framework as it provides harmonized neuroimaging data across broad age ranges and a scientifically grounded rationale for modeling structured variation across age-defined clusters.

\section{Model and Methods}
\label{sec:model}

This section introduces the Multilevel Covariate-Assisted Principal Regression (MCAP) framework. Section~\ref{sub:method} specifies the model and derives the hierarchical likelihood. Section~\ref{sub:algorithm} presents the block coordinate-descent algorithm for parameter estimation and the strategy for identifying multiple components. Section~\ref{sub:inference} describes the two-stage bootstrap procedure for inference on the regression parameters. Section~\ref{sub:asmp} establishes the asymptotic properties of the estimators.

Consider a hierarchically nested dataset with $m$ clusters at the first level. Assume that, in cluster $i$, the number of units is $n_{i}$, for $i=1,\dots,m$. Let $\by_{ijt}\in\mathbb{R}^{p}$, a $p$-dimensional vector, denote the $t$th sample of unit $j$ in cluster $i$, for $j=1,\dots,n_{i}$ and $t=1,\dots,T_{ij}$, where $T_{ij}$ is the number of samples. It is assumed that $\by_{ijt}$ follows a multivariate normal distribution with mean zero and covariance matrix $\bSigma_{ij}$. Without loss of generality, the distribution mean is assumed to be zero to focus on the study of heterogeneity in the covariance matrices. In practice, one can center the data to zero to make the assumption valid. The heterogeneity of the covariance matrices is associated with cluster and/or individual characteristics.
Let $\bx_{1i(j)}\in\mathbb{R}^{q_{1}}$ denote the $q_{1}$-dimensional fixed-effect covariate of cluster $i$ (unit $j$) and $\bx_{2ij}\in\mathbb{R}^{q_{2}}$ denote the $q_{2}$-dimensional random-effect covariate of unit $j$ in cluster $i$. The covariate $\bx_{1i(j)}$ can be either a cluster-level covariate ($\bx_{1i}$) or a unit-level covariate ($\bx_{1ij}$). Denote $\bx_{ij}=(1,\bx_{1i(j)}^\top,\bx_{2ij}^\top)^\top\in\mathbb{R}^{q}$, where $q=1+q_{1}+q_{2}$. Assume that, for each cluster, there exists a linear rotation with unit $\ell_{2}$-norm, denoted as $\bgamma_{i}\in\mathbb{R}^{p}$ and $\|\bgamma_{i}\|_{2}=1$. Under this rotation, the following model is satisfied,
\begin{equation}\label{eq:model}
  \log(\bgamma_{i}^\top\bSigma_{ij}\bgamma_{i}) = \beta_{0}+\bx_{1i(j)}^\top\bbeta_{1}+\bx_{2ij}^\top(\bbeta_{2}+\bvartheta_{i})+\varepsilon_{i},
\end{equation}
where $\bbeta=(\beta_{0},\bbeta_{1}^\top,\bbeta_{2}^\top)^\top\in\mathbb{R}^{q}$ is the model coefficient, $\varepsilon_{i}$ is the random component normally distributed with mean zero and variance $\sigma^{2}$, and $\bvartheta_{i}$ is the random-effect coefficient assumed to be normally distributed with mean zero and covariance matrix $\bOmega\in\mathbb{R}^{q_{2}\times q_{2}}$. Denote $\beta_{0i}\equiv\beta_{0}+\varepsilon_{i}$ and $\bbeta_{2i}\equiv\bbeta_{2}+\bvartheta_{i}$. $\beta_{0i}$ is the random intercept and $\bbeta_{2i}$ is the random slope corresponding to those random-effect covariates.
Because the linear rotation $\bgamma_{i}$ satisfies $\|\bgamma_{i}\|_{2}=1$, a spherical distribution is considered. Specifically, $\bgamma_{i}$ is assumed to follow a von Mises-Fisher distribution with mean direction $\bgamma\in\mathbb{R}^{p}$ ($\|\bgamma\|_{2}=1$) and concentration $\kappa>0$. Throughout the model and likelihood, $\bgamma_{i}$ denotes the Euclidean unit-norm direction. In the computation below, an auxiliary working vector $\tilde{\bgamma}_{i}$ is used to impose a data-adaptive normalization, and the direction entering the likelihood is obtained by rescaling $\tilde{\bgamma}_{i}$ to unit Euclidean norm.


\subsection{Methods}
\label{sub:method}

Under these distributional assumptions, for a multilevel model like~\eqref{eq:model}, one can estimate the parameters by maximizing the hierarchical likelihood. For $\bgamma_{i}$'s, a von Mises-Fisher distribution is assumed. The likelihood function is
\begin{equation}
  f(\bgamma_{i};\bgamma,\kappa)=C_{p}(\kappa)\exp\left(\kappa\bgamma^\top\bgamma_{i}\right),
\end{equation}
where
\begin{equation}
  C_{p}(\kappa)=\frac{\kappa^{p/2-1}}{(2\pi)^{p/2}I_{p/2-1}(\kappa)},
\end{equation}
and $I_{\nu}(\cdot)$ denotes the modified Bessel function of the first kind at order $\nu$. Up to an additive constant, the negative hierarchical likelihood under model~\eqref{eq:model} is
\begin{eqnarray}
  \ell &=& \sum_{i=1}^{m}\sum_{j=1}^{n_{i}}\frac{T_{ij}}{2}\left\{(\beta_{0i}+\bx_{1i(j)}^\top\bbeta_{1}+\bx_{2ij}^\top\bbeta_{2i})+(\bgamma_{i}^\top\bS_{ij}\bgamma_{i})\exp(-\beta_{0i}-\bx_{1i(j)}^\top\bbeta_{1}-\bx_{2ij}^\top\bbeta_{2i})\right\} \nonumber \\
  && +\sum_{i=1}^{m}\left\{\frac{1}{2}\log\sigma^{2}+\frac{(\beta_{0i}-\beta_{0})^{2}}{2\sigma^{2}}\right\}+\sum_{i=1}^{m}\left\{\frac{1}{2}\log|\bOmega|+\frac{1}{2}(\bbeta_{2i}-\bbeta_{2})^\top\bOmega^{-1}(\bbeta_{2i}-\bbeta_{2})\right\} \nonumber \\
  && +\sum_{i=1}^{m}\left\{-\log{C_{p}(\kappa)}-\kappa(\bgamma^\top\bgamma_{i})\right\}, \label{eq:likelihood}
\end{eqnarray}
where $\bS_{ij}=\sum_{t=1}^{T_{ij}}\by_{ijt}\by_{ijt}^\top/T_{ij}$ is the sample covariance matrix. The first part of the function corresponds to the conditional likelihood of the unit-level observations given $\{\bgamma_{i},\beta_{0i},\bbeta_{2i}\}_{i}$, the second part corresponds to the likelihood of the random intercepts ($\beta_{0i}$), the third corresponds to the likelihood of the random slopes ($\bbeta_{2i}$), and the last part corresponds to the likelihood of the random rotations ($\bgamma_{i}$). We use this hierarchical likelihood rather than marginalizing over $\{\bgamma_{i},\beta_{0i},\bbeta_{2i}\}_{i}$ because the resulting marginal likelihood is analytically and computationally inconvenient. In addition, it has been shown that maximizing the hierarchical likelihood function is asymptotically equivalent to maximizing the standard marginalized likelihood function~\citep{lee1996hierarchical}. 
The equivalence result of \citet{lee1996hierarchical} was established for generalized linear mixed models under the condition that the number of observations per cluster grows, that is $T=\min_{i,j}T_{ij}\rightarrow\infty$. As $T\rightarrow\infty$, the posterior distribution of cluster-level random effects concentrates with increasing precision, and the hierarchical likelihood profile over the random effects converges to the marginal likelihood. See~\citet{lee2006generalized} for a comprehensive discussion of this equivalence in the broader class of hierarchical generalized linear models. To avoid degeneration in the linear rotations, $\bgamma_{i}$, the following optimization problem is proposed to estimate model parameters:
\begin{eqnarray}
  \text{minimize} && \ell, \label{eq:opt}  \\
  \text{such that} && \tilde{\bgamma}_{i}^{\top}\bH_{i}\tilde{\bgamma}_{i}=1 \text{ and } \bgamma_{i}=\tilde{\bgamma}_{i}/\|\tilde{\bgamma}_{i}\|_2,\quad \text{for } i=1,\dots,m, \nonumber
\end{eqnarray}
where $\bH_{i}$ is a positive-definite matrix. The auxiliary vector $\tilde{\bgamma}_{i}$ is introduced only for the data-adaptive normalization in the optimization step. The model parameter remains the Euclidean unit-norm direction, $\bgamma_{i}$, which enters the hierarchical likelihood and is modeled by the von Mises-Fisher distribution.
When $\bH_{i}=\matI_{p}$, a $p$-dimensional identity matrix, the working normalization is equivalent to imposing the unit $\ell_{2}$-norm constraint and recovers the standard PCA normalization. For a likelihood-based optimization problem, \citet{krzanowski1984principal} and \citet{zhao2021covariate} demonstrated that incorporating distributional information into the constraint leads to better estimation performance. As the proposed estimation approach is likelihood based, we set $\bH_{i}=\sum_{j=1}^{n_{i}}T_{ij}\bS_{ij}/\sum_{j=1}^{n_{i}}T_{ij}$, the average sample covariance matrix of cluster $i$, for $i=1,\dots,m$.
Setting $\bH_{i}=\matI_{p}$ treats all directions equally, while setting $\bH_{i}$ to the cluster average sample covariance tilts the working feasible set toward directions along which the data already vary.
Geometrically, the constraint, $\tilde{\bgamma}_{i}^{\top}\bH_{i}\tilde{\bgamma}_{i}=1$, defines an ellipsoid aligned with the data cloud rather than a sphere, so the working update corresponds to a generalized eigenvector problem with respect to the data-adaptive metric $\bH_{i}$ rather than the identity (see also supplementary Section~\ref{appendix:sec:alg_mlcap}). After this working update, $\tilde{\bgamma}_{i}$ is rescaled to $\bgamma_{i}=\tilde{\bgamma}_{i}/\|\tilde{\bgamma}_{i}\|_2$, so the von Mises-Fisher model remains on the Euclidean unit sphere. \citet{krzanowski1984principal} and \citet{zhao2021covariate} showed empirically and theoretically that this data-adaptive normalization substantially reduces estimation error for the projection direction in finite samples.
In practice, $\bH_{i}$ is computed once before the main optimization loop and does not change across iterations.

\subsection{Algorithm}
\label{sub:algorithm}

One advantage of using the hierarchical likelihood is computational convenience. For the likelihood function in~\eqref{eq:likelihood}, a block coordinate-descent algorithm can be used to obtain the solution. Algorithm~\ref{alg:mlcap} summarizes the estimation procedure and supplementary Section~\ref{appendix:sec:alg_mlcap} provides more details. For $\{\beta_{0i}\}$, $\bbeta_{1}$, and $\{\bbeta_{2i}\}$, the updates are obtained using the Newton-Raphson method. For the hyperparameters, $(\beta_{0},\sigma^{2})$ and $(\bbeta_{2},\bOmega)$, explicit solutions for the updates are provided below by minimizing the negative hierarchical likelihood function~\eqref{eq:likelihood} as
\[
  \beta_{0}^{(s+1)}=\frac{1}{m}\sum_{i=1}^{m}\beta_{0i}^{(s+1)}, \quad \sigma^{2(s+1)}=\frac{1}{m}\sum_{i=1}^{m}(\beta_{0i}^{(s+1)}-\beta_{0}^{(s+1)})^{2};
\]
\[
  \bbeta_{2}^{(s+1)}=\frac{1}{m}\sum_{i=1}^{m}\bbeta_{2i}^{(s+1)}, \quad \bOmega^{(s+1)}=\frac{1}{m}\sum_{i=1}^{m}\left(\bbeta_{2i}^{(s+1)}-\bbeta_{2}^{(s+1)}\right)\left(\bbeta_{2i}^{(s+1)}-\bbeta_{2}^{(s+1)}\right)^\top.
\]
For $\bgamma_{i}$, $i=1,\dots,m$, the update is obtained by first solving for the working direction $\tilde{\bgamma}_{i}$,
\begin{eqnarray}
  \text{minimize} && \tilde{\bgamma}_{i}^\top\left\{\sum_{j=1}^{n_{i}}\frac{T_{ij}}{2}\exp(-\beta_{0i}^{(s+1)}-\bx_{1i(j)}^{\top}\bbeta_{1}^{(s+1)}-\bx_{2ij}^\top\bbeta_{2i}^{(s+1)})\bS_{ij} \right\}\tilde{\bgamma}_{i}-\kappa^{(s)}\bgamma^{(s)\top}\tilde{\bgamma}_{i}, \label{eq:opt_gamma} \\
  \text{such that} && \tilde{\bgamma}_{i}^\top\bH_{i}\tilde{\bgamma}_{i}=1. \nonumber
\end{eqnarray}
The unit-norm direction is then set to $\bgamma_{i}^{(s+1)}=\tilde{\bgamma}_{i}^{(s+1)}/\|\tilde{\bgamma}_{i}^{(s+1)}\|_2$. Details of solving the working problem are presented in Section~\ref{appendix:sec:alg_mlcap} of the supplementary materials.
For $\bgamma$ and $\kappa$, let $\bar{\bgamma}=\sum_{i=1}^{m}\bgamma_{i}/m$ and $\bar{R}=\|\bar{\bgamma}\|$, where the average is computed from the unit-norm directions. The updates are
\[
  \bgamma^{(s+1)}=\frac{\bar{\bgamma}^{(s+1)}}{\bar{R}^{(s+1)}}, \quad \kappa^{(s+1)}=\frac{\bar{R}^{(s+1)}(p-\bar{R}^{2(s+1)})}{1-\bar{R}^{2(s+1)}}.
\]
This update is the commonly used approximation to the maximum likelihood update for the von Mises-Fisher concentration parameter, as discussed in~\citet{sra2012short}.
A generic scenario is considered in Algorithm~\ref{alg:mlcap}, where both fixed-effect covariates ($\bx_{1i(j)}$) and random-effect covariates ($\bx_{2ij}$) are present. In the case that the variances in $\bbeta_{2i}$ degenerate to zero, that is, constant slopes are assumed for $\bx_{2ij}$, one can combine $\bx_{1i(j)}$ and $\bx_{2ij}$ and drop corresponding terms from the likelihood function. Analogously, if all covariates have a random slope, terms involving $\bbeta_{1}$ should be removed from the likelihood function.

The block coordinate-descent scheme updates the objective iteratively, and convergence is declared when the relative change in the objective falls below a tolerance across consecutive iterations.
Because the objective function is non-convex in $(\{\bgamma_{i}\},\{\beta_{0i}\},\bbeta_{1},\{\bbeta_{2i}\})$ jointly, the algorithm may converge to a local minimum. To mitigate this, multiple random initializations are used and the solution with the smallest objective value is retained, as indicated in Step~5 of Algorithm~\ref{alg:mlcap}. The dominant computational cost per iteration is the generalized eigenvalue decomposition for each $\bgamma_{i}$ update, which is $O(p^{3})$; the remaining Newton-Raphson and closed-form updates are $O(p^{2})$ per cluster. The total cost per iteration is therefore $O(mp^{3} + M_{n}p^{2})$, where $M_{n}=\sum_{i}\sum_{j}T_{ij}$, which is in general feasible for moderate dimensions, for example, $p \leq 100$.

\begin{algorithm}
  \caption{\label{alg:mlcap}The algorithm for optimization problem~\eqref{eq:opt}.}
  \begin{algorithmic}[1]
    \INPUT $\{(\by_{ij1},\dots,\by_{ijT_{ij}}),\bx_{ij} \mid j=1,\dots,n_{i},~i=1,\dots,m\}$

    \State \textbf{initialization}: $(\{\bgamma_{i}^{(0)}\},\{\beta_{0i}^{(0)}\},\bbeta_{1}^{(0)},\{\bbeta_{2i}^{(0)}\},\beta_{0}^{(0)},\sigma^{2(0)},\bbeta_{2}^{(0)},\bOmega^{(0)},\bgamma^{(0)},\kappa^{(0)})$

    \Repeat \; for iteration $s=0,1,2,\dots$

    \State \; update parameters by solving~\eqref{eq:opt}, denoted as $\{\bgamma_{i}^{(s+1)}\}$, $\{\beta_{0i}^{(s+1)}\}$, $\bbeta_{1}^{(s+1)}$, $\{\bbeta_{2i}^{(s+1)}\}$, $\beta_{0}^{(s+1)}$, $\sigma^{2(s+1)}$, $\bbeta_{2}^{(s+1)}$, $\bgamma^{(s+1)}$, and $\bOmega^{(s+1)},\kappa^{(s+1)}$, respectively,

    \Until{the objective function in~\eqref{eq:opt}} converges;

    \State consider a random series of initializations, repeat Steps 1--4, and choose the results with the minimum objective value.

    \OUTPUT $(\{\hat{\bgamma}_{i}\}, \{\hat{\beta}_{0i}\},\hat{\bbeta}_{1}, \{\hat{\bbeta}_{2i}\},\hat{\beta}_{0},\hat{\sigma}^{2},\hat{\bbeta}_{2},\hat{\bOmega},\hat{\bgamma},\hat{\kappa})$
  \end{algorithmic}
\end{algorithm}

Algorithm~\ref{alg:mlcap} presents the algorithm for identifying the first linear rotation (component). To identify higher-order components, a similar strategy is followed as in~\citet{zhao2021covariate}. Suppose~$(k-1)$ components have been identified, denoted as $\hat{\bGamma}_{i}^{(k-1)}=(\hat{\bgamma}_{i}^{(1)},\dots,\hat{\bgamma}_{i}^{(k-1)})\in\mathbb{R}^{p\times (k-1)}$, for $i=1,\dots,m$ and $k\geq 2$. Algorithm~\ref{alg:mlcap_higher} summarizes the steps of identifying the $k$th component. Step 1 is to first remove the identified components from the outcome data, $\bY_{ij}=(\by_{ij1},\dots,\by_{ijT_{ij}})$. The rank of the resulting data, $\hat{\bY}_{ij}$, is at most $p-(k-1)$, that is, the last $(k-1)$ singular values in $\bD_{ij}^{(k)}$ are zeros. The sample covariance matrices are then rank deficient. Since the optimization is likelihood based, to make the computation tractable, a rank-completion step, Step 3 in Algorithm~\ref{alg:mlcap_higher}, is introduced, where the estimated random intercepts, $\{\hat{\beta}_{0i}^{(1)},\dots,\hat{\beta}_{0i}^{(k-1)}\}$, are used to replace those zero singular values. The resulting data, $\tilde{\bY}_{ij}^{(k)}$, are then used as input to Algorithm~\ref{alg:mlcap} to identify the $k$th component, $\hat{\bgamma}_{i}^{(k)}$. An orthogonal constraint would require the newly identified $\hat{\bgamma}_{i}^{(k)}$ to be orthogonal to the matrix $\hat{\bGamma}_{i}^{(k-1)}$. However, this orthogonal constraint significantly reduces the feasible set and increases the computation complexity. In practice, removing the identified components from the data enforces the desired orthogonality in the deflated representation without imposing the constraint explicitly.

\begin{algorithm}
  \caption{\label{alg:mlcap_higher}The algorithm of identifying the $k$th component for $k\geq 2$.}
  \begin{algorithmic}[1]
    \INPUT $\{\{\bY_{ij}=(\by_{ij1},\dots,\by_{ijT_{ij}}),\bx_{ij}\}_{j}, \hat{\bGamma}_{i}^{(k-1)}, \{\hat{\beta}_{0i}^{(1)},\dots,\hat{\beta}_{0i}^{(k-1)}\} \mid i=1,\dots,m\}$

    \State For $j=1,\dots,n_{i}$, let $\hat{\bY}_{ij}^{(k)}=\bY_{ij}-\bY_{ij}\hat{\bGamma}_{i}^{(k-1)}\hat{\bGamma}_{i}^{(k-1)\top}$.

    \State Apply singular value decomposition on $\hat{\bY}_{ij}^{(k)}$ such that $\hat{\bY}_{ij}^{(k)}=\bU_{ij}^{(k)}\bD_{ij}^{(k)}\bV_{ij}^{(k)\top}$.

    \State Let $\tilde{\bY}_{ij}^{(k)}=\bU_{ij}^{(k)}\tilde{\bD}_{ij}^{(k)}\bV_{ij}^{(k)\top}$ with 
      \[
        \tilde{\bD}_{ij}^{(k)}=\diag\left\{D_{ij1}^{(k)},\dots,D_{ij(p-(k-1))}^{(k)},\sqrt{\exp(\hat{\beta}_{0i}^{(1)})T_{ij}},\dots,\sqrt{\exp(\hat{\beta}_{0i}^{(k-1)})T_{ij}} \right\}.
      \]

    \State Use $\tilde{\bY}_{ij}^{(k)}$ as the input and run Algorithm~\ref{alg:mlcap}.

    \OUTPUT $(\{\hat{\bgamma}_{i}^{(k)}\}, \{\hat{\beta}_{0i}^{(k)}\},\hat{\bbeta}_{1}^{(k)}, \{\hat{\bbeta}_{2i}^{(k)}\},\hat{\beta}_{0}^{(k)},\hat{\sigma}^{2(k)},\hat{\bbeta}_{2}^{(k)},\hat{\bOmega}^{(k)},\hat{\bgamma}^{(k)},\hat{\kappa}^{(k)})$
  \end{algorithmic}
\end{algorithm}

To determine the number of components, the metric of average deviation from diagonality (DfD) proposed in \citet{zhao2021covariate} is adapted to the considered multilevel data setting. Let $\hat{\bGamma}_{i}^{(k)}$ denote the $k$th identified component of cluster $i$, for $i=1,\dots,m$. The metric is defined as 
\begin{equation}
  \text{DfD}(k)=\prod_{i=1}^{m}\prod_{j=1}^{n_{i}}\left[\frac{\det\left\{\diag(\hat{\bGamma}_{i}^{(k)\top}\bS_{ij}\hat{\bGamma}_{i}^{(k)})\right\}}{\det(\hat{\bGamma}_{i}^{(k)\top}\bS_{ij}\hat{\bGamma}_{i}^{(k)})} \right]^{T_{ij}/\sum_{i}\sum_{j}T_{ij}},
\end{equation}
where $\diag(\bA)$ is a diagonal matrix taking the diagonal elements of a square matrix $\bA$ and $\det(\bA)$ is the determinant of $\bA$. For all $i$, if $\hat{\bGamma}_{i}^{(k)}$ diagonalizes all $\bS_{ij}$'s, for $j=1,\dots,n_{i}$, $\text{DfD}(k)=1$, otherwise the value is greater than one based on Hadamard's inequality. In practice, one can set a threshold, for example $\text{DfD}(k)\leq 2$, to determine $k$. This threshold was suggested by \citet{zhao2021covariate} based on simulation studies.

\subsection{Inference}
\label{sub:inference}

In this study, inference focuses on quantifying the uncertainty of regression coefficients, $(\beta_{0},\bbeta_{1},\bbeta_{2})$. Inference on the linear rotations, $\{\bgamma_{i}\}$ and $\bgamma$, is beyond the scope, as these carry inherent sign ambiguity and are treated as structural components for the purpose of regression inference.

A two-stage bootstrap procedure is proposed. The key motivation is the separation of convergence rates established in Theorem~\ref{thm:asym_known}, where the regression parameters converge at rates determined by $m$ and $M_{n}$, while the $\hat{\bgamma}_{i}$'s are consistent (Theorem~\ref{thm:consistency}). When $m/(nT)\rightarrow 0$, the estimation error in $\hat{\bgamma}_{i}$ is asymptotically negligible relative to that of the regression parameters (as shown in the proof of Theorem~\ref{thm:asym_known}). It is therefore valid to fix $\{\hat{\bgamma}_{i}\}$ at their full-data estimates and bootstrap only the regression sub-problem. This avoids re-running the computationally expensive $\bgamma_{i}$ estimation in each bootstrap replicate and matching estimated directions across replicates.
The following by-cluster resampling strategy respects the hierarchical data structure and correctly propagates the dominant source of variability. Bootstrap standard errors and confidence intervals are then obtained from the bootstrap estimates. 

\begin{description}
  \item[Stage 0 (Full estimation).] Run Algorithm~\ref{alg:mlcap} on the full dataset to obtain $\left(\{\hat{\bgamma}_{i}\},\{\hat{\beta}_{0i}\},\hat{\bbeta}_{1},\{\hat{\bbeta}_{2i}\},\right.\\
      \left.\hat{\beta}_{0},\hat{\sigma}^{2},\hat{\bbeta}_{2},\hat{\bOmega},\hat{\bgamma},\hat{\kappa}\right)$.
    Compute the projected sample variances
    \[
      \hat{s}_{ij}=\hat{\bgamma}_{i}^{\top}\bS_{ij}\hat{\bgamma}_{i}, \quad j=1,\dots,n_{i},\;i=1,\dots,m.
    \]
    The estimates $\{\hat{\bgamma}_{i}\}$ and $\{\hat{s}_{ij}\}$ are held fixed throughout the bootstrap.

  \item[Stage 1 (By-cluster bootstrap).] For each replicate $b=1,\dots,B$:
    \begin{enumerate}[(a)]
      \item Resample $m$ clusters with replacement to form a bootstrap sample $\{i_{1}^{*},\dots,i_{m}^{*}\}$.
      \item For each resampled cluster $i_{k}^{*}$, resample $n_{i_{k}^{*}}$ units within the cluster.
      \item Minimize the reduced hierarchical likelihood with $\{\bgamma_{i}\}$ and $\bgamma$ fixed to obtain bootstrap estimates $(\hat{\beta}_{0}^{(b)},\hat{\bbeta}_{1}^{(b)},\hat{\bbeta}_{2}^{(b)},\hat{\sigma}^{2(b)},\hat{\bOmega}^{(b)})$.
    \end{enumerate}
\end{description}


\subsection{Asymptotic properties}
\label{sub:asmp}

This section studies the asymptotic properties of the proposed estimators under two settings: (i) $\{\bgamma_{i}\}$ and $\bgamma$ are known, and (ii) they are unknown and need to be estimated. Following the definition above, $M_{n}=\sum_{i}\sum_{j}T_{ij}$, $T=\min_{i,j}T_{ij}$, $n=\min_{i}n_{i}$, and let a superscript $*$ denote the true parameter value.

\textit{Known $\{\bgamma_{i}\}$ and $\bgamma$.}
When $\{\bgamma_{i}\}$ and $\bgamma$ are given, the projected sample variance $s_{ij}=\bgamma_{i}^{\top}\bS_{ij}\bgamma_{i}$ is directly observable and the model reduces to a log-linear mixed effects model. The asymptotic normality of the MLEs can be derived.

\begin{theorem}
  \label{thm:asym_known}
  Assume $p<T=\min_{i,j}T_{ij}$ and $p$ is fixed. Let $M_{n}$, $T$, and $n$ be as defined above. Suppose $m^{-1}\sum_{i}n_{i}^{-1}\sum_{j}\bx_{1i(j)}\bx_{1i(j)}^{\top}\rightarrow\bQ$ as $m\rightarrow\infty$, where $\bQ$ is positive definite. Define the profile information matrix
  \begin{equation}
    \label{eq:profile_info}
    \mathcal{J}_{n}=\bH_{11}-\sum_{i=1}^{m}\bC_{i}\bD_{i}^{-1}\bC_{i}^{\top},
  \end{equation}
  where $\bH_{11}=\sum_{i,j}(T_{ij}/2)\bx_{1i(j)}\bx_{1i(j)}^{\top}\in\mathbb{R}^{q_{1}\times q_{1}}$, and for each cluster $i=1,\dots,m$,
  \[
    \bC_{i}=\sum_{j=1}^{n_{i}}\frac{T_{ij}}{2}\bx_{1i(j)}\!\left(1,\bx_{2ij}^{\top}\right)\in\mathbb{R}^{q_{1}\times(1+q_{2})},
  \]
  \[
    \bD_{i}=\begin{pmatrix} 
      N_{i}/2+1/\sigma^{2*} & \sum_{j}(T_{ij}/2)\bx_{2ij}^{\top}\\[4pt] \sum_{j}(T_{ij}/2)\bx_{2ij} & \sum_{j}(T_{ij}/2)\bx_{2ij}\bx_{2ij}^{\top}+\bOmega^{*-1}\end{pmatrix}\in\mathbb{R}^{(1+q_{2})\times(1+q_{2})},
  \]
  with $N_{i}=\sum_{j}T_{ij}$. Then, as $m,T\rightarrow\infty$,
  \begin{enumerate}[(i)]
    \item $\hat{\bbeta}_{1}$ is asymptotically normal with $\mathcal{J}_{n}^{1/2}\!\left(\hat{\bbeta}_{1}-\bbeta_{1}^{*}\right)\overset{\mathcal{D}}{\longrightarrow}\mathcal{N}(\matzero,\matI_{q_{1}})$.
  \end{enumerate}
  Furthermore, if $m/(nT)\rightarrow 0$ as $m,n,T\rightarrow\infty$, then
  \begin{enumerate}[(i)]
    \setcounter{enumi}{1}
    \item $\displaystyle\sqrt{m}\left(\hat{\beta}_{0}-\beta_{0}^{*}\right)\overset{\mathcal{D}}{\longrightarrow}\mathcal{N}(0,\sigma^{2*})$;
    \item $\displaystyle\sqrt{m}\left(\hat{\sigma}^{2}-\sigma^{2*}\right)\overset{\mathcal{D}}{\longrightarrow}\mathcal{N}(0,2\sigma^{4*})$;
    \item $\displaystyle\sqrt{m}\left(\hat{\bbeta}_{2}-\bbeta_{2}^{*}\right)\overset{\mathcal{D}}{\longrightarrow}\mathcal{N}(\matzero,\bOmega^{*})$;
    \item $\displaystyle\sqrt{m}\,\mathrm{vec}\!\left(\hat{\bOmega}-\bOmega^{*}\right)\overset{\mathcal{D}}{\longrightarrow}\mathcal{N}\{\matzero,(\matI_{q_{2}^{2}}+\bK)(\bOmega^{*}\otimes\bOmega^{*})\}$, where $\matI_{q_{2}^{2}}$ is the $q_{2}^{2}$-dimensional identity matrix and $\bK$ is the $q_{2}^{2}\times q_{2}^{2}$ commutation matrix.
  \end{enumerate}
\end{theorem}

The proof of the theorem is presented in supplementary Section~\ref{appendix:proof:thm_asym_known}. The profile information $\mathcal{J}_{n}$ in~\eqref{eq:profile_info} is the Schur complement~\citep{horn1990matrix} of the $\bbeta_{1}$ block in the joint Hessian of the hierarchical likelihood with respect to $(\{\beta_{0i}\},\bbeta_{1},\{\bbeta_{2i}\})$, which corresponds to the profile likelihood information for $\bbeta_{1}$ after concentrating out the cluster-level effects~\citep{murphy2000profile}. It accounts for the coupling between the fixed-effect estimate $\hat{\bbeta}_{1}$ and the jointly estimated random effects $\hat{\beta}_{0i}$ and $\hat{\bbeta}_{2i}$. Results (ii)--(v) mirror the classical marginal asymptotic distributions of sample means and sample covariance matrices for Gaussian random variables. The condition $m/(nT)\rightarrow 0$ ensures that the estimation error in the cluster-level effects is negligible relative to the between-cluster variation.

\textit{Unknown $\{\bgamma_{i}\}$ and $\bgamma$.}
When $\{\bgamma_{i}\}$ and $\bgamma$ are unknown and need to be estimated, deriving asymptotic normality is substantially more complex due to the $\ell_{2}$-norm constraint and the non-linear dependence of the likelihood on $\bgamma_{i}$. We establish consistency under an identifiability condition.

\begin{assumption}
  \label{assum:ident}
  Let $\bar{\ell}_{n}(\bTheta)=M_{n}^{-1}\ell(\bTheta)$ denote the normalized negative hierarchical likelihood, where $\bTheta$ collects all regression parameters, variance components, and direction parameters, with a fixed sign convention for the directions. The following conditions hold.
  \begin{description}
    \item[(A1)] The parameter space for $\bTheta$ is compact.
    \item[(A2)] There exists a deterministic population objective $\bar{\ell}(\bTheta)=\operatorname*{plim}\bar{\ell}_{n}(\bTheta)$ such that $\bar{\ell}_{n}$ converges uniformly in probability to $\bar{\ell}$ in a neighborhood of $\bTheta^{*}$.
    \item[(A3)] The true parameter $\bTheta^{*}$ is the unique minimizer of $\bar{\ell}(\bTheta)$.
    \item[(A4)] For each fixed cluster $i$, the direction block is locally identifiable. Specifically, let $N_{i}=\sum_{j=1}^{n_{i}}T_{ij}$ and $\mu_{ij}^{*}=\beta_{0i}^{*}+\bx_{1i(j)}^{\top}\bbeta_{1}^{*}+\bx_{2ij}^{\top}\bbeta_{2i}^{*}$. With
      \[
        \bar{\bA}_{i}^{*}=\lim_{n,T\rightarrow\infty}\frac{1}{2N_{i}}\sum_{j=1}^{n_{i}}T_{ij}\exp(-\mu_{ij}^{*})\bSigma_{ij},
        \qquad
        \bH_{i}^{*}=\lim_{n,T\rightarrow\infty}\frac{1}{N_{i}}\sum_{j=1}^{n_{i}}T_{ij}\bSigma_{ij},
      \]
      the smallest generalized eigenvalue of $(\bar{\bA}_{i}^{*},\bH_{i}^{*})$ is simple. If $\tilde{\bgamma}_{i}^*$ denotes the corresponding generalized eigenvector under the chosen sign convention, then $\tilde{\bgamma}_{i}^*/\|\tilde{\bgamma}_{i}^*\|_2=\bgamma_{i}^*$.
  \end{description}
\end{assumption}

\begin{theorem}
  \label{thm:consistency}
  Let $\hat{\bTheta}$ be a global minimizer of $\bar{\ell}_{n}(\bTheta)$. Under Assumption~\ref{assum:ident} and the conditions of Theorem~\ref{thm:asym_known}, as $m,n,T\rightarrow\infty$,
  \[
    \hat{\bgamma}_{i}\overset{\mathcal{P}}{\longrightarrow}\bgamma_{i}^{*}\;\;\text{for each fixed }i,\quad\hat{\bgamma}\overset{\mathcal{P}}{\longrightarrow}\bgamma^{*},\quad\hat{\kappa}\overset{\mathcal{P}}{\longrightarrow}\kappa^{*},
  \]
  \[
    \hat{\bbeta}_{1}\overset{\mathcal{P}}{\longrightarrow}\bbeta_{1}^{*},\quad\hat{\beta}_{0}\overset{\mathcal{P}}{\longrightarrow}\beta_{0}^{*},\quad\hat{\sigma}^{2}\overset{\mathcal{P}}{\longrightarrow}\sigma^{2*},\quad\hat{\bbeta}_{2}\overset{\mathcal{P}}{\longrightarrow}\bbeta_{2}^{*},\quad\hat{\bOmega}\overset{\mathcal{P}}{\longrightarrow}\bOmega^{*}.
  \]
  The consistency statement for $\hat{\kappa}$ assumes that the von Mises-Fisher concentration is obtained from the exact likelihood equation; the closed-form update in Algorithm~\ref{alg:mlcap} is a numerical approximation to this inverse.
\end{theorem}

The proof of Theorem~\ref{thm:consistency} is in supplementary Section~\ref{appendix:proof:thm_consistency}. Assumption~\ref{assum:ident} is an identifiability condition on the full population objective. The generalized-eigenvalue separation in the direction block is a convenient sufficient condition ensuring that each cluster-specific projection is locally identifiable. It does not require all covariance matrices $\{\bSigma_{ij}\}_{j=1}^{n_{i}}$ within a cluster to share a common eigenbasis. A common within-cluster eigenstructure is sufficient and gives a direct principal component interpretation, but the weaker requirement for consistency is uniqueness of the population minimizer, or equivalently adequate separation of the weighted covariance contrast defining $\bar{\bA}_{i}^{*}$ relative to $\bH_{i}^{*}$.


\section{Simulation Study}
\label{sec:sim}

This section evaluates the performance of the proposed approach through simulation studies. In the study, the covariance matrices are generated following the eigendecomposition, $\bSigma_{ij}=\bPi_{i}\bLambda_{ij}\bPi_{i}^\top$, where $\bPi_{i}=(\bpi_{i1},\dots,\bpi_{ip})\in\mathbb{R}^{p\times p}$ is an orthonormal matrix and $\bLambda_{ij}=\diag\{\lambda_{ij1},\dots,\lambda_{ijp}\}\in\mathbb{R}^{p\times p}$ is a diagonal matrix, for $i=1,\dots,m$ and $j=1,\dots,n_{i}$. $\bpi_{ik}$ and $\lambda_{ijk}$ are corresponding eigenvector and eigenvalue, for $k=1,\dots,p$. Two dimensions, D2 and D4, are considered to satisfy the model assumptions and targets to be identified. For these two dimensions, the eigenvectors are generated from von Mises-Fisher distributions with mean direction, $\bpi_{2}$ and $\bpi_{4}$, respectively, and concentration parameter $\kappa$, and the corresponding eigenvalues are generated following the proposed model~\eqref{eq:model}. For the rest dimensions, the eigenvectors are generated from the null space of $\mathrm{span}\{\bpi_{i2},\bpi_{i4}\}$ to ensure the orthogonality and the eigenvalues are randomly generated from a log-normal distribution.
For the regression model~\eqref{eq:model}, three covariates are considered ($q=3+1=4$), where two are fixed-effect covariates ($q_{1}=2$) and one is a random-effect covariate ($q_{2}=1$). One fixed-effect covariate is generated from a Bernoulli distribution with probability $0.5$ to be one, and the other fixed-effect and random-effect covariates are generated from a normal distribution with mean zero and variance $0.5^2$. For all dimensions, the intercept, $\beta_{0}$ exponentially decays from $5$ to $-3$. For D2, $\bbeta_{1}=(1,-0.5)^\top$ and $\beta_{2}=-0.5$; for D4, $\bbeta_{1}=(-1,0.5)^\top$ and $\beta_{2}=0.5$. $\varepsilon_{i}$ and $\vartheta_{i}$ are independently generated from a normal distribution with mean zero and variance $0.1^{2}$. With $\bSigma_{ij}$, $\by_{ijt}$ is generated from a multivariate normal distribution with mean zero and covariance matrix $\bSigma_{ij}$.

Two approaches are considered to identify the covariate-related components and estimate the corresponding model coefficients, the proposed multilevel covariate-assisted principal regression (denoted as \textbf{MCAP}) and an approach derived from the single-level framework proposed in \citet{zhao2021covariate} (denoted as \textbf{SCAP}). The SCAP approach has two steps: (i) apply the covariance regression framework in \citet{zhao2021covariate} to each cluster separately; and (ii) average the estimated linear projections and model coefficients over clusters. For both approaches, the number of components is chosen as $\max\{k:\mathrm{DfD}(\hat{\bGamma}^{(k)})<2\}$. The performance using this criterion was studied in \citet{zhao2021covariate} and a threshold of two was recommended. We therefore do not repeat the evaluation of component-number selection in this study.

For evaluation, two dimension settings of $p=5$ and $p=20$ are considered. The number of clusters is set to $m=20$. $n_{i}$ and $T_{ij}$ are generated from Poisson distributions with mean $n=100,500$ and $T=100,500$, respectively. 
Table~\ref{table:sim} presents the simulation results for one component (D4). For the competing method, SCAP, only results with $n=T=100$ for both dimensions are presented for demonstration. From the table, when $p=5$, the proposed MCAP approach yields a good estimate of the projections and model coefficients. As the sample size increases, the performance improves. Although the SCAP approach performs slightly better in estimating $\bgamma$, the bias in estimating the model coefficients is higher. When the dimension increases to $p=20$, the estimate from SCAP is inaccurate. For the MCAP approach, although the bias in estimating $\bgamma$ is high with $n=T=100$, when the sample sizes increase to $n=T=500$, the performance is much improved. These results suggest the superior performance of the proposed multilevel approach.
Direction similarities are computed using the final Euclidean unit-norm directions, consistent with the von Mises-Fisher model for $\bgamma_{i}$.
Section~\ref{appendix:sub:sim_inference} in the supplementary materials presents the performance of the proposed approaches for inference on the regression coefficients.

\begin{table}
  \caption{\label{table:sim}Performance of the simulation study, including the similarity of $\hat{\bgamma}$ and $\bpi_{j}$ and the estimated standard deviation (SE), bias, and mean squared error (MSE) in estimating $\beta$. Data dimension $p=5,20$, number of clusters $m=20$, sample sizes $n=100,500$ and $T=100,500$.}
  \begin{center}
    \begin{tabular}{l l l l r c r r c r r}
      \hline
      & & & & \multicolumn{1}{c}{$\hat{\bgamma}$} && \multicolumn{2}{c}{$\hat{\beta}_{11}$} && \multicolumn{2}{c}{$\hat{\beta}_{2}$} \\
      \cline{5-5}\cline{7-8}\cline{10-11}
      \multicolumn{1}{c}{\multirow{-2}{*}{$p$}} & \multicolumn{1}{c}{\multirow{-2}{*}{$n$}} & \multicolumn{1}{c}{\multirow{-2}{*}{$T$}} & \multicolumn{1}{c}{\multirow{-2}{*}{Method}} & \multicolumn{1}{c}{$|\langle \hat{\bgamma},\bpi_{j}\rangle|$ (SE)} && \multicolumn{1}{c}{Bias} & \multicolumn{1}{c}{MSE} && \multicolumn{1}{c}{Bias} & \multicolumn{1}{c}{MSE} \\
      \hline
      & & & SCAP & $0.983$ ($0.011$) && $1.503$ & $2.437$ && $-0.234$ & $0.185$ \\
      & \multirow{-2}{*}{$100$} & \multirow{-2}{*}{$100$} & MCAP & $0.921$ ($0.092$) && $0.246$ & $0.160$ && $-0.166$ & $0.040$ \\
      \cline{2-11}
      \multirow{-3}{*}{$5$} & $500$ & $500$ & MCAP & $0.967$ ($0.008$) && $0.082$ & $0.021$ && $-0.052$ & $0.003$ \\
      \hline
      & & & SCAP & $0.276$ ($0.055$) && $1.243$ & $1.650$ && $0.018$ & $0.001$ \\
      & \multirow{-2}{*}{$100$} & \multirow{-2}{*}{$100$} & MCAP & $0.577$ ($0.056$) && $0.102$ & $0.014$ && $0.016$ & $0.001$ \\
      \cline{2-11}
      \multirow{-3}{*}{$20$} & $500$ & $500$ & MCAP & $0.819$ ($0.026$) && $0.101$ & $0.125$ && $-0.067$ & $0.008$ \\
      \hline
    \end{tabular}
  \end{center}
\end{table}

\section{The Human Connectome Project Lifespan Study}
\label{sec:hcp}

We apply the proposed approach to the HCP Lifespan Study, including HCP-Development (HCP-D), HCP-Young Adult (HCP-YA), and HCP-Aging (HCP-A), covering ages from five to ninety. The three cohorts together form a hierarchically nested dataset, with subjects nested within age-defined clusters. Age clusters are defined with a five-year window, resulting in $m=17$ clusters. $1,500$ unrelated subjects ($815$ females and $685$ males) are analyzed and the sample size within each cluster~($n_{i}$) is presented in Table~\ref{appendix:table:hcp_sample}. For each subject, resting-state fMRI time courses are extracted from $p=75$ regions of interest (ROIs) using the Harvard-Oxford Atlas in FSL~\citep{smith2004advances}. The atlas includes $60$ cortical and $15$ subcortical regions, which are grouped into ten functional modules. For HCP-D and HCP-A subjects, the number of time points is $T_{ij}=229$ and for HCP-YA subjects, $T_{ij}=590$. Age, Age$^{2}$, and Age$^{3}$ are included as fixed-effect covariates to capture the non-linear lifespan profile of FC, and sex is included as a random effect to capture the sex difference.
To better interpret the identified components, $\bgamma_{i}$'s and $\bgamma$ are sparsified following an \textit{ad hoc} procedure using a fused lasso regression~\citep{tibshirani2005sparsity} to encourage spatial smoothness across regions and constancy within the same functional module.

Using the proposed DfD criterion, the method identifies a single dominant component, suggesting that age- and sex-related heterogeneity in the FC matrices is concentrated along one principal component. This result implies that the major lifespan-related variation in FC can be summarized parsimoniously by a one-dimensional subnetwork score, estimated by borrowing information across age-defined clusters. 
For this component, the estimated regression coefficients, together with their $95\%$ bootstrap confidence intervals, are presented in Table~\ref{table:hcp} and the estimated FC profile is visualized in Figure~\ref{fig:hcp_agetraj_brain_C1}. The projected FC score declines from childhood and adolescence into the late $20$s, rises steadily through midlife and late adulthood, peaks around the early $70$s, and then declines again from approximately ages $70$ to $90$. This overall shape is shared by both males and females, with males showing higher projected FC values than females across the age range. This observation of sex difference in brain connectivity is in line with existing literature~\citep{zhang2016sex,zhang2024different}. The shape of the aging curve points to a structured decline--increase--decline pattern along the identified subnetwork. This type of nonlinear life-course organization is broadly consistent with recent studies showing that functional connectivity is dynamically reorganized across development, adulthood, and aging rather than changing in a strictly linear pattern~\citep{sun2025human,yang2025lifelong,li2025effective}.

An especially notable feature of the estimated trajectory is the partial mirroring between later life and early development. The descending segment from approximately ages $70$ to $90$ resembles, in direction and broad curvature, the earlier decline observed from childhood into young adulthood. In this sense, the late-life trajectory partially echoes the developmental trajectory, with the FC score moving away from its late-adulthood peak toward levels closer to those seen earlier in life. This mirroring phenomenon suggests that the identified subnetwork is particularly sensitive to brain reorganization at both ends of the lifespan, first during developmental refinement and later during aging-related decline or dedifferentiation. Recent studies of functional gradients, spatial variability, and effective connectivity across the lifespan similarly support nonlinear late-life reconfiguration and age-related dedifferentiation of brain systems~\citep{wang2024aging,shen2025dedifferentiation,yang2025lifelong}.

The chord diagram in Figure~\ref{fig:hcp_avgbrain_chord} further elaborates on this interpretation by showing how the spatial pattern of the identified subnetwork varies across age clusters. There are relatively few strong similarities among subnetworks from age bins below $50$, whereas subnetworks from later adulthood, especially from roughly age $50$ onward, are more frequently linked to one another. This suggests that the identified subnetwork becomes more homogeneous across neighboring age bins in later life, even though the FC score itself continues to change over that period. In other words, later-adulthood subnetworks appear to share a more conserved spatial organization than those observed during childhood, adolescence, and early-to-mid adulthood, where the pattern is more heterogeneous. This suggests that aging-related reorganization may proceed along a more common spatial axis, whereas earlier life stages show greater variability in how that axis is expressed. Such a pattern is consistent with recent evidence for both substantial inter-individual heterogeneity in brain aging and systematic network shifts in later life~\citep{dohm2024middle}.

At the regional level, the loading profile (Figure~\ref{appendix:fig:hcp_avgbrain} in the supplementary materials) indicates that this component is driven most strongly by the rostral middle frontal gyrus~(rMFG), with additional prominent contributions from the middle temporal gyrus~(MTG), the inferior parietal cortex~(IPC), and the isthmus cingulate gyrus~(ICG), as well as the supramarginal gyrus~(SMG) with a slightly lower contribution. The component spans across multiple known brain systems, including frontoparietal, default mode network (DMN), ventral-attention, and lateral temporal regions. 
The dominant frontal and temporal contributions are consistent with recent evidence that age-related functional changes are especially pronounced in the association cortex and control-related systems, including frontoparietal and default-mode circuitry~\citep{pan2024brain,khalilian2024age,stanford2024age}.
The strong contribution of rMFG is consistent with the well-established role of prefrontal systems in executive control and working memory, while the MTG has been repeatedly implicated in semantic and language-related processing. The SMG is plausibly relevant given its established role in phonological processing and verbal working memory. The ICG and IPC contributions are also notable, as posterior midline and inferior parietal areas are core components of the DMN and association cortex systems that show robust age-related FC alterations.

On average across all age clusters, greater differences are observed in the visual, somato-motor, ventral-attention, fronto-parietal and DMN networks (Figure~\ref{fig:hcp_bubble_plot}). This finding indicates that lifespan-related heterogeneity is concentrated not only in higher-order association networks such as the DMN and fronto-parietal control systems, but also in sensory and sensorimotor systems. The involvement of the DMN and fronto-parietal network is consistent with recent work showing that association cortex is especially sensitive to age-related changes in functional organization, including reduced network segregation, altered inter-network interactions, and greater variability in older adults \citep{pan2024brain,yang2025lifelong}. The contribution of the ventral-attention network further suggests age-related modulation of salience and attentional reorienting systems, which have been implicated in both developmental maturation and late-life decline~\citep{shen2025dedifferentiation}. At the same time, the pronounced differences observed in the visual and somato-motor networks indicate that primary and unimodal systems also undergo substantial lifespan reconfiguration, in line with recent connectome-wide studies showing nonlinear age effects in sensory and sensorimotor circuitry~\citep{orlichenko2024somatomotor,sun2025human}. Taken together, our finding suggests that the dominant lifespan-related FC component is not restricted to a single network class, but instead reflects coordinated changes across both association and sensory-motor systems.

The anatomical configuration of these regions further supports an interpretation centered on distributed higher-order cognition rather than isolated regional change. Frontal, parietal, and temporal association cortices are interconnected by major long-range white matter fiber bundles, especially the superior longitudinal fasciculus and arcuate fasciculus complex, which link perisylvian frontal, parietal, and temporal territories and are widely implicated in language and executive processing~\citep{mousley2025topological}. Although the present analysis does not impose any anatomical or network structure a priori, the resulting component aligns with a plausible fronto-parietal-temporal association network.

The proposed MCAP framework is well suited to uncovering lifespan-related variation in both brain organization and within-network FC. By modeling the cluster-specific projection direction, $\bgamma_{i}$, as a random effect on the unit sphere under a von Mises-Fisher distribution, MCAP integrates information across age clusters to estimate a population-level mean direction, $\hat{\bgamma}$, while allowing principled between-cluster heterogeneity in subnetwork orientation. This multilevel structure is critical in the present setting, where a single-level analysis conducted separately within each age cluster (such as SCAP) would have limited power to detect covariate effects within smaller clusters and averaging of cluster-specific estimates would not properly account for between-cluster variability or propagate uncertainty. The bootstrap confidence intervals further provide finite-sample uncertainty quantification. Together, these methodological features enable the findings in the HCP study, including the nonlinear within-network FC profile, the sex difference, the convergence of neural directions in late adulthood, and the identification of a distributed fronto-parietal-temporal subnetwork.

\begin{table}
  \caption{\label{table:hcp}Estimated regression coefficients and $95\%$ bootstrap confidence intervals (CIs) from $500$ bootstrap samples. SE: standard error.}
  \begin{center}
    \begin{tabular}{l r r c}
      \hline
      \multicolumn{1}{c}{Covariate} & \multicolumn{1}{c}{Estimate} & \multicolumn{1}{c}{SE} & \multicolumn{1}{c}{$95\%$ CI} \\
      \hline
      Age & $0.511$ & $0.198$ & $(0.123, 0.898)$ \\
      Age$^{2}$ & $0.732$ & $0.107$ & $(0.522, 0.941)$ \\
      Age$^{3}$ & $-0.358$ & $0.057$ & $(-0.470, -0.245)$ \\
      Male & $0.210$ & $0.100$ & $(0.013, 0.407)$ \\
      \hline
    \end{tabular}
  \end{center}
\end{table}

\begin{figure}
  \begin{center}
    \includegraphics[width=\textwidth]{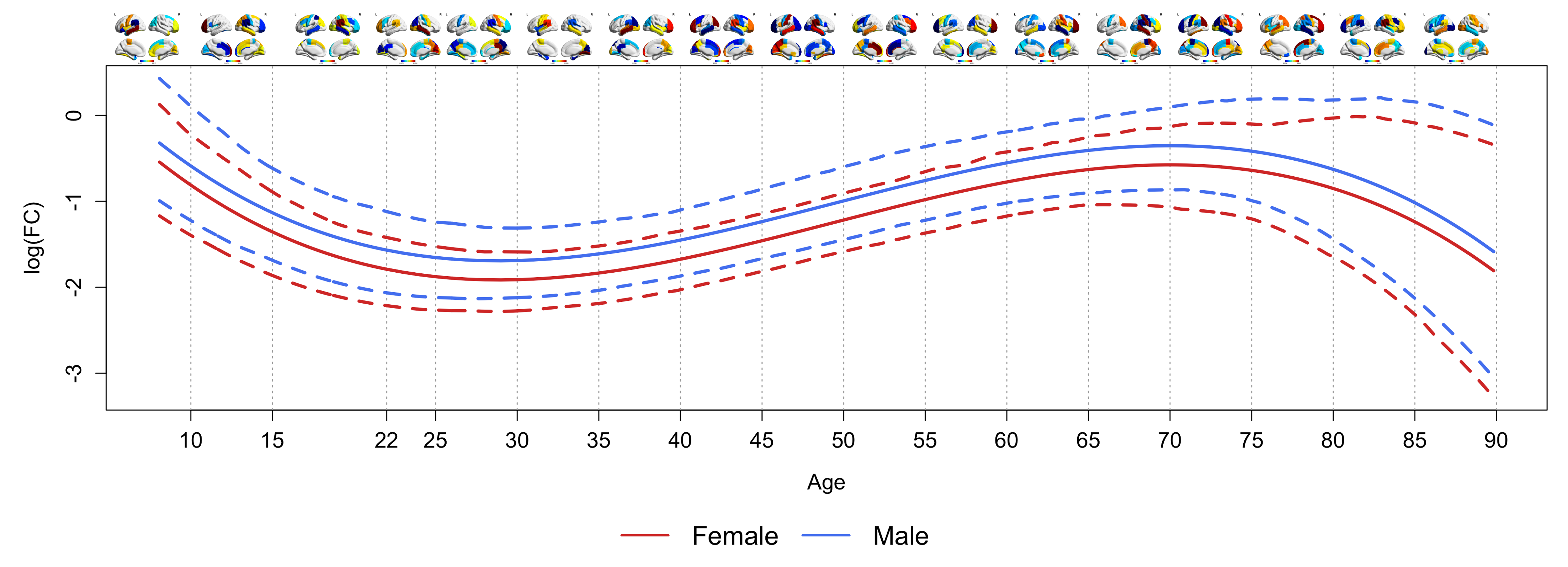}
  \end{center}
  \caption{\label{fig:hcp_agetraj_brain_C1}The profile of functional connectivity within the identified brain subnetwork over age, as well as the brain map of each age cluster, in the HCP study.}
\end{figure}

\begin{figure}
  \subfloat[\label{fig:hcp_avgbrain_chord}]{\includegraphics[width=0.45\textwidth]{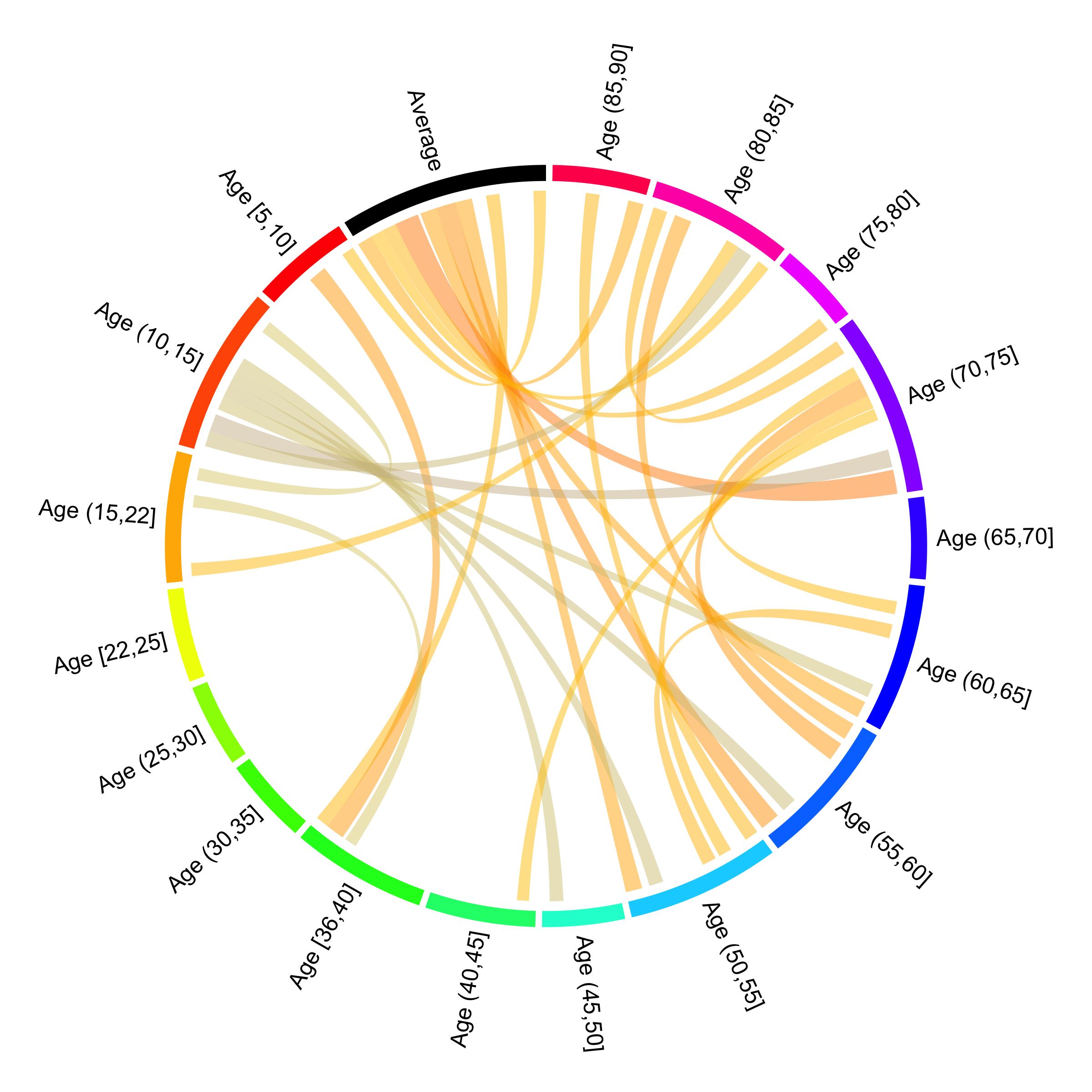}}
  \enskip{}
  \subfloat[\label{fig:hcp_avgbrain_brain}]{\includegraphics[width=0.45\textwidth]{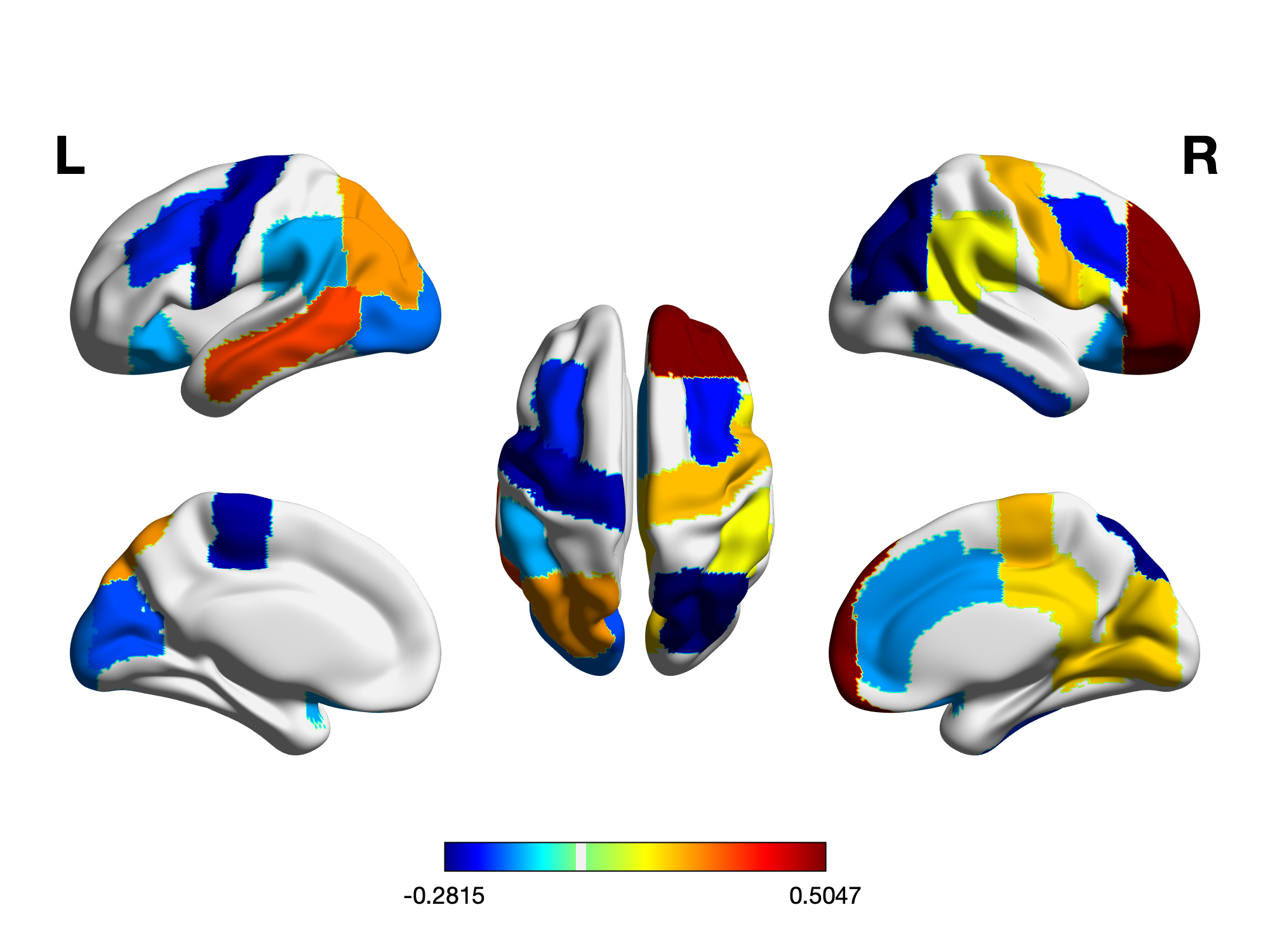}}

  \subfloat[\label{fig:hcp_bubble_plot}]
  {\includegraphics[width=1\textwidth]{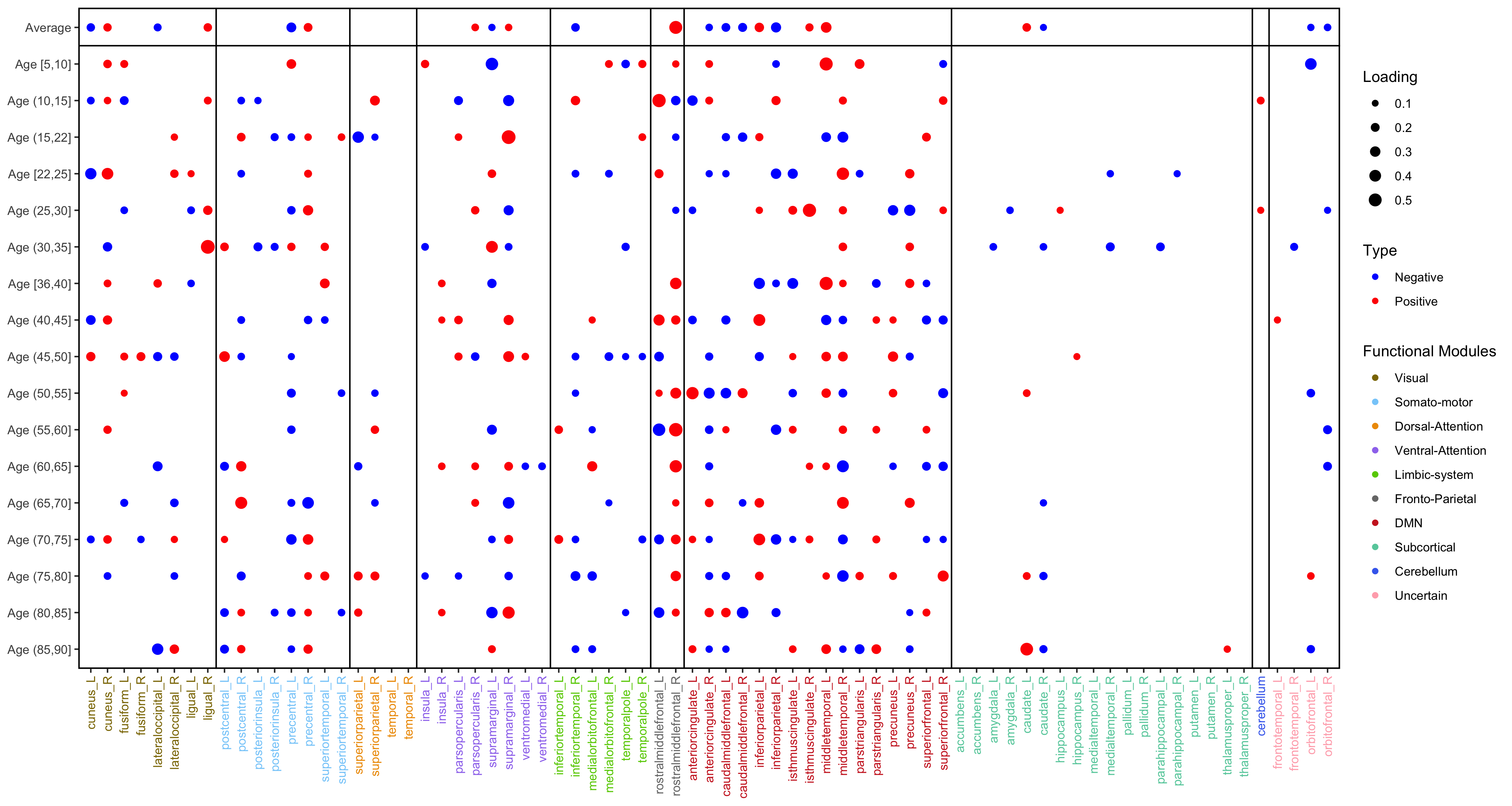}}
  
  \caption{\label{fig:hcp_avgbrain}Result of the average projection of the component identified in the HCP study. (a) Similarity of the identified projections for each age cluster and the average projection. (b) The brain map of the average projection. (c) The bubble plot of the projection by age and functional module.}
\end{figure}

\section{Discussion}
\label{sec:discussion}

In this paper, we propose a multilevel covariate-assisted principal (MCAP) regression framework for hierarchically nested covariance matrix outcomes. The proposed MCAP model extends the single-level CAP framework~\citep{zhao2021covariate} by allowing the covariate-associated projection direction to vary across clusters and by modeling these cluster-specific directions through a von Mises-Fisher distribution. Under the normality assumption, a hierarchical likelihood is introduced to jointly estimate model parameters. A block coordinate-descent algorithm is developed for computation, and a two-stage bootstrap procedure is proposed for inference on the regression parameters. We establish asymptotic normality for the regression parameters when the projection directions are known and consistency of the estimators when the projection directions are unknown. Simulation studies demonstrate that the proposed approach improves the estimation of model coefficients and projection directions compared with a single-level analysis, especially when the data dimension is moderate to high.

The proposed framework provides a useful tool for studying lifespan variation in brain functional connectivity from the HCP Lifespan Study. Existing approaches typically either vectorize the covariance matrix, model each age group separately, or impose a common projection direction across all clusters. These strategies make it difficult to borrow information across age groups while still allowing age-specific variation in the covariance structure. MCAP addresses this limitation by estimating age-cluster-specific projections while simultaneously learning a population-level direction and regression effects. In the HCP analysis, this allows the functional connectivity profile across ages five to ninety to be summarized through a dominant covariate-associated brain subnetwork, while still preserving heterogeneity in the subnetworks across age clusters. The fitted model reveals a nonlinear lifespan trajectory of functional connectivity, a sex difference along the identified direction, and greater similarity of the brain subnetworks among older age groups. The identified subnetwork involves regions in fronto-parietal, default mode, ventral-attention, and lateral temporal systems, including the SMG, rMFG, and MTG, suggesting a distributed cross-network pattern that would be difficult to characterize by analyzing each age cluster independently or by using a fixed projection direction for all clusters.

Several limitations remain. First, the current model assumes that the observations are Gaussian so that the sample covariance matrices provide a likelihood-based summary. In applications where this assumption is violated, quasi-likelihood or robust estimating equation versions of the proposed framework may be useful. Second, the von Mises-Fisher model provides a parsimonious description of the random projection directions, but it is unimodal. Mixtures of spherical distributions could be considered when the directions form multiple clusters. Third, the current theory establishes consistency of the projection directions but does not provide asymptotic normality or direct inference for $\bgamma_{i}$'s and $\bgamma$. Developing such inference would require accounting for the constraint on the directions and the sign indeterminacy of eigenvector-type estimators. Fourth, higher-order components are identified sequentially through deflation, which is computationally convenient but may be less efficient than a joint multi-component procedure. Future work may also consider sparse or structured projections for high-dimensional neuroimaging data, time-varying projection directions for longitudinal cluster processes, formal tests for the number of components, and extensions to non-Gaussian covariance or dependence summaries.



\clearpage

\appendix
\counterwithin{figure}{section}
\counterwithin{table}{section}
\counterwithin{equation}{section}
\counterwithin{lemma}{section}
\counterwithin{theorem}{section}

\section*{Appendix}

This Appendix collects the technical proof of the theorems in the main text, additional theoretical results, and additional data analysis results.

\section{Theory and Proof}
\label{appendix:sec:proof}

\subsection{Proof of Theorem~\ref{thm:asym_known}}
\label{appendix:proof:thm_asym_known}

\begin{proof}

We prove each part in turn. Throughout, let $\mu_{ij}^{*}=\beta_{0i}^{*}+\bx_{1i(j)}^{\top}\bbeta_{1}^{*}+\bx_{2ij}^{\top}\bbeta_{2i}^{*}$, $s_{ij}=\bgamma_{i}^{\top}\bS_{ij}\bgamma_{i}$, and $r_{ij}=s_{ij}\exp(-\mu_{ij}^{*})$. Since $T_{ij}r_{ij}\sim\chi^{2}_{T_{ij}}$, we have $\mathbb{E}r_{ij}=1$ and $\mathrm{Var}(r_{ij})=2/T_{ij}$, and the $r_{ij}$'s are mutually independent.

\textit{Part (i).}
The joint score of the negative hierarchical likelihood with respect to $\bbeta_{1}$, $\{\beta_{0i}\}$, and $\{\bbeta_{2i}\}$, evaluated at the true parameter values, is
\begin{equation}
  S_{n}(\bbeta_{1}^{*}) = \frac{\partial\ell}{\partial\bbeta_{1}}\bigg|_{*} = \sum_{i=1}^{m}\sum_{j=1}^{n_{i}}\frac{T_{ij}}{2}\left(1-r_{ij}\right)\bx_{1i(j)},
  \label{eq:score_beta1}
\end{equation}
and for each cluster $i$,
\[
  \frac{\partial\ell}{\partial\beta_{0i}}\bigg|_{*}=\sum_{j}\frac{T_{ij}}{2}(1-r_{ij})+\frac{\varepsilon_{i}}{\sigma^{2*}},\qquad
  \frac{\partial\ell}{\partial\bbeta_{2i}}\bigg|_{*}=\sum_{j}\frac{T_{ij}}{2}(1-r_{ij})\bx_{2ij}+\bOmega^{*-1}\bvartheta_{i},
\]
where $\varepsilon_{i}=\beta_{0i}^{*}-\beta_{0}^{*}$ and $\bvartheta_{i}=\bbeta_{2i}^{*}-\bbeta_{2}^{*}$. At convergence, the joint score at $(\hat{\bbeta}_{1},\{\hat{\beta}_{0i}\},\{\hat{\bbeta}_{2i}\})$ is zero. Since $\ell$ is the negative hierarchical likelihood, its Hessian blocks are positive. A first-order Taylor expansion around the true values gives
\begin{equation}
  \matzero\approx\frac{\partial\ell}{\partial\bbeta_{1}}\bigg|_{*}+\bH_{11}\left(\hat{\bbeta}_{1}-\bbeta_{1}^{*}\right)+\sum_{i=1}^{m}\bC_{i}\begin{pmatrix}\hat{\beta}_{0i}-\beta_{0i}^{*}\\[2pt]\hat{\bbeta}_{2i}-\bbeta_{2i}^{*}\end{pmatrix},
  \label{eq:taylor_joint}
\end{equation}
where the Hessian blocks $\bH_{11}$ and $\bC_{i}$ are evaluated at $r_{ij}\approx 1$ (as $T\rightarrow\infty$) and are as defined in~\eqref{eq:profile_info}. Likewise, the score equations for $\beta_{0i}$ and $\bbeta_{2i}$ give
\[
  \begin{pmatrix}\hat{\beta}_{0i}-\beta_{0i}^{*}\\[2pt]\hat{\bbeta}_{2i}-\bbeta_{2i}^{*}\end{pmatrix}\approx\bD_{i}^{-1}\!\left[\sum_{j}\frac{T_{ij}}{2}(r_{ij}-1)\begin{pmatrix}1\\\bx_{2ij}\end{pmatrix}-\begin{pmatrix}\varepsilon_{i}/\sigma^{2*}\\\bOmega^{*-1}\bvartheta_{i}\end{pmatrix}-\bC_{i}^{\top}(\hat{\bbeta}_{1}-\bbeta_{1}^{*})\right],
\]
where $\bD_{i}$ is as in~\eqref{eq:profile_info}. Substituting into~\eqref{eq:taylor_joint} and solving for $\hat{\bbeta}_{1}-\bbeta_{1}^{*}$ yields the profile-score representation:
\begin{equation}
  \hat{\bbeta}_{1}-\bbeta_{1}^{*}\approx\mathcal{J}_{n}^{-1}\bw_{n},
  \label{eq:beta1_profile}
\end{equation}
where $\mathcal{J}_{n}=\bH_{11}-\sum_{i}\bC_{i}\bD_{i}^{-1}\bC_{i}^{\top}$ is the profile information matrix defined in~\eqref{eq:profile_info}, and $\bw_{n}$ is the profile score
\[
  \bw_{n}=-S_{n}(\bbeta_{1}^{*})-\sum_{i=1}^{m}\bC_{i}\bD_{i}^{-1}\!\left[\sum_{j}\frac{T_{ij}}{2}(r_{ij}-1)\begin{pmatrix}1\\\bx_{2ij}\end{pmatrix}-\begin{pmatrix}\varepsilon_{i}/\sigma^{2*}\\\bOmega^{*-1}\bvartheta_{i}\end{pmatrix}\right].
\]
Since the $r_{ij}$'s are mutually independent with $\mathbb{E}(r_{ij})=1$ and $\mathrm{Var}(r_{ij})=2/T_{ij}$, and $\varepsilon_{i}$, $\bvartheta_{i}$ are independent of the $r_{ij}$'s, a direct calculation gives $\mathbb{E}(\bw_{n})=\matzero$ and $\mathrm{Var}(\bw_{n})=\mathcal{J}_{n}$. The Lindeberg condition for the multivariate CLT is satisfied under bounded covariates, so $\mathcal{J}_{n}^{-1/2}\bw_{n}\overset{\mathcal{D}}{\longrightarrow}\mathcal{N}(\matzero,\matI_{q_{1}})$. Combining with~\eqref{eq:beta1_profile} and Slutsky's theorem gives $\mathcal{J}_{n}^{1/2}(\hat{\bbeta}_{1}-\bbeta_{1}^{*})\overset{\mathcal{D}}{\longrightarrow}\mathcal{N}(\matzero,\matI_{q_{1}})$.

\textit{Parts (ii) and (iii).}
For each cluster $i$, conditional on the population fixed effects, the score for $\beta_{0i}$ is $\sum_{j}T_{ij}(1-r_{ij})/2+(\beta_{0i}-\beta_{0})/\sigma^{2}$, and the corresponding observed Fisher information is $\sum_{j}T_{ij}r_{ij}/2+1/\sigma^{2}$. Since $r_{ij}\overset{\mathcal{P}}{\longrightarrow}1$, the information is approximately $N_{i}/2+1/\sigma^{2}$, where $N_{i}=\sum_{j}T_{ij}$. The joint profiling over $\bbeta_{1}$ contributes terms linear in $\hat{\bbeta}_{1}-\bbeta_{1}^{*}$. Under the centered fixed-effect parametrization for the intercept, these terms average to $o_{p}(m^{-1/2})$ in the hyperparameter updates. Therefore, by a standard Newton-Raphson argument,
\begin{equation}
  \hat{\beta}_{0i}=\beta_{0i}^{*}+O_{p}(N_{i}^{-1/2}).
  \label{eq:b0i_rate}
\end{equation}
Since $\hat{\beta}_{0}=m^{-1}\sum_{i}\hat{\beta}_{0i}$ and $\beta_{0i}^{*}=\beta_{0}^{*}+\varepsilon_{i}$ with $\varepsilon_{i}\sim\mathcal{N}(0,\sigma^{2*})$,
\[
  \sqrt{m}\left(\hat{\beta}_{0}-\beta_{0}^{*}\right)=m^{-1/2}\sum_{i=1}^{m}\varepsilon_{i}+O_{p}\!\left(\sqrt{\frac{m}{nT}}\right).
\]
Under $m/(nT)\rightarrow 0$, the second term is $o_{p}(1)$, and the CLT applied to $\{\varepsilon_{i}\}$ gives $\sqrt{m}(\hat{\beta}_{0}-\beta_{0}^{*})\overset{\mathcal{D}}{\longrightarrow}\mathcal{N}(0,\sigma^{2*})$.

For $\hat{\sigma}^{2}=m^{-1}\sum_{i}(\hat{\beta}_{0i}-\hat{\beta}_{0})^{2}$, using~\eqref{eq:b0i_rate},
\[
  \hat{\beta}_{0i}-\hat{\beta}_{0}=(\varepsilon_{i}-\bar{\varepsilon})+O_{p}(N_{i}^{-1/2}),
\]
where $\bar{\varepsilon}=m^{-1}\sum_{i}\varepsilon_{i}$. Under $m/(nT)\rightarrow 0$, $\hat{\sigma}^{2}\approx m^{-1}\sum_{i}(\varepsilon_{i}-\bar{\varepsilon})^{2}$, which is the MLE of the variance of random variables from $\mathcal{N}(0,\sigma^{2*})$. By the CLT, $\sqrt{m}(\hat{\sigma}^{2}-\sigma^{2*})\overset{\mathcal{D}}{\longrightarrow}\mathcal{N}(0,2\sigma^{4*})$.

\textit{Parts (iv) and (v).}
The same profiling argument applies to the random slopes. Under the centered fixed-effect parametrization, the first-order contribution of $\hat{\bbeta}_{1}-\bbeta_{1}^{*}$ to the empirical mean and covariance of the estimated random slopes is $o_{p}(m^{-1/2})$. Thus, with $\bvartheta_{i}=\bbeta_{2i}^{*}-\bbeta_{2}^{*}\sim\mathcal{N}(\matzero,\bOmega^{*})$, we have $\sqrt{m}(\hat{\bbeta}_{2}-\bbeta_{2}^{*})\overset{\mathcal{D}}{\longrightarrow}\mathcal{N}(\matzero,\bOmega^{*})$. For $\hat{\bOmega}=m^{-1}\sum_{i}(\hat{\bbeta}_{2i}-\hat{\bbeta}_{2})(\hat{\bbeta}_{2i}-\hat{\bbeta}_{2})^{\top}$, this is asymptotically equivalent to the sample covariance of random vectors from $\mathcal{N}(\matzero,\bOmega^{*})$. Standard results for the sample covariance matrix yield $\sqrt{m}\,\mathrm{vec}(\hat{\bOmega}-\bOmega^{*})\overset{\mathcal{D}}{\longrightarrow}\mathcal{N}\{\matzero,(\matI_{q_{2}^{2}}+\bK)(\bOmega^{*}\otimes\bOmega^{*})\}$, where $\bK$ is the $q_{2}^{2}\times q_{2}^{2}$ commutation matrix. 

\end{proof}

The profile information $\mathcal{J}_{n}$ depends on the covariate structure. Two limiting cases are instructive. (a) If $\bx_{1i(j)}$ has no between-cluster variation, i.e., $\sum_{j}T_{ij}\bx_{1i(j)}=N_{i}\bar{\bx}_{1i}$ with $\bar{\bx}_{1i}=\matzero$ for all $i$, the cross-Hessian $\bC_{i}$ vanishes and $\mathcal{J}_{n}=\bH_{11}\sim M_{n}\bQ/2$, so $\mathrm{Var}(\hat{\bbeta}_{1})\sim 2\bQ^{-1}/M_{n}$, converging at the rate of $\sqrt{M_{n}}$. (b) If $\bx_{1i(j)}=\bar{\bx}_{1i}$ is constant within each cluster and $q_{2}=0$, the Schur complement reduces to $\mathcal{J}_{n}\rightarrow m\bQ_{B}/\sigma^{2*}$ as $T\rightarrow\infty$, where $\bQ_{B}=m^{-1}\sum_{i}\bar{\bx}_{1i}\bar{\bx}_{1i}^{\top}$, giving $\mathrm{Var}(\hat{\bbeta}_{1})\sim\sigma^{2*}\bQ_{B}^{-1}/m$ at the rate of $\sqrt{m}$. In general, $\mathcal{J}_{n}$ interpolates between these extremes depending on the within- and between-cluster covariate variation.

\subsection{Proof of Theorem~\ref{thm:consistency}}
\label{appendix:proof:thm_consistency}

\begin{proof}
We prove consistency by applying the argmin theorem to the full normalized objective, rather than by treating the direction and regression blocks separately. Let $\bTheta$ collect all unknown parameters in Algorithm~\ref{alg:mlcap}, and let $\bar{\ell}_{n}(\bTheta)=M_{n}^{-1}\ell(\bTheta)$. By Assumption~\ref{assum:ident}, the parameter space is compact, $\bar{\ell}_{n}$ converges uniformly in probability to a deterministic population objective $\bar{\ell}$, and $\bar{\ell}$ has the unique minimizer $\bTheta^*$. Based on the result in \citet{vaart1998asymptotic},
\[
  \hat{\bTheta}\overset{\mathcal{P}}{\longrightarrow}\bTheta^*.
\]
Taking the relevant coordinates yields consistency of the regression parameters and variance components.

It remains only to clarify the direction block. For fixed $i$, the sample working-direction update is driven by
\[
  \bar{\bA}_{i}=\frac{1}{2N_{i}}\sum_{j=1}^{n_{i}}T_{ij}\exp(-\hat{\mu}_{ij})\bS_{ij},\qquad
  \bH_{i}=\frac{1}{N_{i}}\sum_{j=1}^{n_{i}}T_{ij}\bS_{ij}.
\]
Uniform convergence of the full objective implies $\hat{\mu}_{ij}$ is evaluated in a shrinking neighborhood of $\mu_{ij}^*$, and the law of large numbers gives $\bar{\bA}_{i}\overset{\mathcal{P}}{\longrightarrow}\bar{\bA}_{i}^{*}$ and $\bH_{i}\overset{\mathcal{P}}{\longrightarrow}\bH_{i}^*$. The von Mises-Fisher linear term in the unnormalized $\tilde{\bgamma}_{i}$ update is $O(1)$ whereas the quadratic term is $O(N_{i})$. Equivalently, after division by $N_{i}$, the linear term is $O(N_{i}^{-1})$ and is asymptotically negligible for the cluster-specific direction. The simple generalized-eigenvalue condition in Assumption~\ref{assum:ident}, together with continuity of isolated eigenvectors, then gives $\hat{\tilde{\bgamma}}_{i}\overset{\mathcal{P}}{\longrightarrow}\tilde{\bgamma}_{i}^*$ for each fixed $i$, under the chosen sign convention. Since Euclidean normalization is continuous away from zero,
\[
  \hat{\bgamma}_{i}=\hat{\tilde{\bgamma}}_{i}/\|\hat{\tilde{\bgamma}}_{i}\|_2
  \overset{\mathcal{P}}{\longrightarrow}
  \tilde{\bgamma}_{i}^*/\|\tilde{\bgamma}_{i}^*\|_2=\bgamma_{i}^*.
\]

Finally, since the $\bgamma_{i}^*$'s are independent von Mises-Fisher random vectors with mean direction $\bgamma^*$ and concentration $\kappa^*$, the sample resultant based on the consistently estimated directions satisfies
\[
  \bar{\bgamma}=m^{-1}\sum_{i=1}^{m}\hat{\bgamma}_{i}\overset{\mathcal{P}}{\longrightarrow}
  \mathbb{E}(\bgamma_{i})
  =A_p(\kappa^*)\bgamma^*,
\]
where $A_p(\kappa)=I_{p/2}(\kappa)/I_{p/2-1}(\kappa)$. Thus $\hat{\bgamma}=\bar{\bgamma}/\|\bar{\bgamma}\|\overset{\mathcal{P}}{\longrightarrow}\bgamma^*$. If $\kappa$ is obtained by the exact von Mises-Fisher likelihood equation $A_p(\hat{\kappa})=\|\bar{\bgamma}\|$, continuity of $A_p^{-1}$ gives $\hat{\kappa}\overset{\mathcal{P}}{\longrightarrow}\kappa^*$. The closed-form update in Algorithm~\ref{alg:mlcap} is the usual numerical approximation to this inverse; its limit is the corresponding deterministic approximation unless the exact inverse is used in the final update. 
\end{proof}

\section{Details of Algorithm~\ref{alg:mlcap}}
\label{appendix:sec:alg_mlcap}

This section presents the computational details for Algorithm~\ref{alg:mlcap}. For $\{\beta_{0i}\}$, $\bbeta_{1}$, and $\{\bbeta_{2i}\}$, the updates can be obtained using the Newton-Raphson method. For $\beta_{0i}$,
\[
  \beta_{0i}^{(s+1)}=\beta_{0i}^{(s)}-\frac{\partial\ell/\partial\beta_{0i}^{(s)}}{\partial^{2}\ell/\partial\beta_{0i}^{(s)2}},
\]
where
\begin{eqnarray*}
  \frac{\partial\ell}{\partial\beta_{0i}} &=& \sum_{j=1}^{n_{i}}\frac{T_{ij}}{2}\left\{1-(\bgamma_{i}^\top\bS_{ij}\bgamma_{i})\exp(-\beta_{0i}-\bx_{1i(j)}^\top\bbeta_{1}-\bx_{2ij}^\top\bbeta_{2i}) \right\}+\frac{\beta_{0i}-\beta_{0}}{\sigma^{2}}, \\
  \frac{\partial^{2}\ell}{\partial\beta_{0i}^{2}} &=& \sum_{j=1}^{n_{i}}\frac{T_{ij}}{2}\left\{(\bgamma_{i}^\top\bS_{ij}\bgamma_{i})\exp(-\beta_{0i}-\bx_{1i(j)}^\top\bbeta_{1}-\bx_{2ij}^\top\bbeta_{2i}) \right\}+\frac{1}{\sigma^{2}}.
\end{eqnarray*}
For $\bbeta_{1}$,
\[
  \bbeta_{1}^{(s+1)}=\bbeta_{1}^{(s)}-\left(\frac{\partial^{2}\ell}{\partial\bbeta_{1}^{(s)}\partial\bbeta_{1}^{(s)\top}} \right)^{-1}\frac{\partial\ell}{\partial\bbeta_{1}^{(s)}},
\]
where
\begin{eqnarray*}
  \frac{\partial\ell}{\partial\bbeta_{1}} &=& \sum_{i=1}^{m}\sum_{j=1}^{n_{i}}\frac{T_{ij}}{2}\left\{\bx_{1i(j)}-(\bgamma_{i}^\top\bS_{ij}\bgamma_{i})\exp(-\beta_{0i}-\bx_{1i(j)}^\top\bbeta_{1}-\bx_{2ij}^\top\bbeta_{2i})\bx_{1i(j)} \right\}, \\
  \frac{\partial^{2}\ell}{\partial\bbeta_{1}\partial\bbeta_{1}^\top} &=& \sum_{i=1}^{m}\sum_{j=1}^{n_{i}}\frac{T_{ij}}{2}\left\{(\bgamma_{i}^\top\bS_{ij}\bgamma_{i})\exp(-\beta_{0i}-\bx_{1i(j)}^\top\bbeta_{1}-\bx_{2ij}^\top\bbeta_{2i})\bx_{1i(j)}\bx_{1i(j)}^\top\right\}.
\end{eqnarray*}
For $\bbeta_{2i}$,
\[
  \bbeta_{2i}^{(s+1)}=\bbeta_{2i}^{(s)}-\left(\frac{\partial^{2}\ell}{\partial\bbeta_{2i}^{(s)}\partial\bbeta_{2i}^{(s)\top}} \right)^{-1}\frac{\partial\ell}{\partial\bbeta_{2i}^{(s)}},
\]
where
\begin{eqnarray*}
  \frac{\partial\ell}{\partial\bbeta_{2i}} &=& \sum_{j=1}^{n_{i}}\frac{T_{ij}}{2}\left\{\bx_{2ij}-(\bgamma_{i}^\top\bS_{ij}\bgamma_{i})\exp(-\beta_{0i}-\bx_{1i(j)}^\top\bbeta_{1}-\bx_{2ij}^\top\bbeta_{2i})\bx_{2ij} \right\}+\bOmega^{-1}(\bbeta_{2i}-\bbeta_{2}), \\
  \frac{\partial^{2}\ell}{\partial\bbeta_{2i}\partial\bbeta_{2i}^\top} &=& \sum_{j=1}^{n_{i}}\frac{T_{ij}}{2}\left\{(\bgamma_{i}^\top\bS_{ij}\bgamma_{i})\exp(-\beta_{0i}-\bx_{1i(j)}^\top\bbeta_{1}-\bx_{2ij}^\top\bbeta_{2i})\bx_{2ij}\bx_{2ij}^\top \right\}+\bOmega^{-1}.
\end{eqnarray*}
For the working direction $\tilde{\bgamma}_{i}$, let
\[
  \bA_{i}=\sum_{j=1}^{n_{i}}\frac{T_{ij}}{2}\exp(-\beta_{0i}-\bx_{1i(j)}^\top\bbeta_{1}-\bx_{2ij}^\top\bbeta_{2i})\bS_{ij},
\]
solving the working problem in~\eqref{eq:opt_gamma} amounts to solving
\begin{eqnarray*}
  \text{minimize} && \tilde{\bgamma}_{i}^\top\bA_{i}\tilde{\bgamma}_{i}-\kappa\bgamma^\top\tilde{\bgamma}_{i}, \\
  \text{such that} && \tilde{\bgamma}_{i}^\top\bH_{i}\tilde{\bgamma}_{i}=1.
\end{eqnarray*}
The Lagrangian form is
\[
  \mathcal{L}_{i}=\tilde{\bgamma}_{i}^\top\bA_{i}\tilde{\bgamma}_{i}-\kappa\bgamma^\top\tilde{\bgamma}_{i}+\lambda_{i}(\tilde{\bgamma}_{i}^\top\bH_{i}\tilde{\bgamma}_{i}-1).
\]
\[
  \frac{\partial\mathcal{L}_{i}}{\partial\tilde{\bgamma}_{i}}=2(\bA_{i}+\lambda_{i}\bH_{i})\tilde{\bgamma}_{i}-\kappa\bgamma=\matzero \quad \Rightarrow \quad \tilde{\bgamma}_{i}=\frac{\kappa}{2}(\bA_{i}+\lambda_{i}\bH_{i})^{-1}\bgamma
\]
\[
  \frac{\partial\mathcal{L}_{i}}{\partial\lambda_{i}}=\tilde{\bgamma}_{i}^\top\bH_{i}\tilde{\bgamma}_{i}-1=0
\]
When $\kappa=0$, the linear term vanishes and the problem reduces to minimizing $\tilde{\bgamma}_{i}^\top\bA_{i}\tilde{\bgamma}_{i}$ subject to $\tilde{\bgamma}_{i}^\top\bH_{i}\tilde{\bgamma}_{i}=1$, which is a generalized eigenvalue problem: the solution is the generalized eigenvector of $(\bA_{i},\bH_{i})$ corresponding to the smallest generalized eigenvalue. When $\kappa>0$, the linear term $-\kappa\bgamma^\top\tilde{\bgamma}_{i}$ perturbs the quadratic problem, and the KKT solution is no longer a single eigenvector in general. Nevertheless, the generalized eigenvectors of $(\bA_{i},\bH_{i})$ form a complete $\bH_{i}$-orthonormal basis for $\mathbb{R}^{p}$, each satisfying the constraint, and provide a natural set of candidate solutions.

Specifically, let $\bxi_{01},\dots,\bxi_{0p}$ be the eigenvectors of the symmetric matrix $\bH_{i}^{-1/2}\bA_{i}\bH_{i}^{-1/2}$, and define $\bxi_{s}=\bH_{i}^{-1/2}\bxi_{0s}$, for $s=1,\dots,p$. Then $\bA_{i}\bxi_{s}=d_{s}\bH_{i}\bxi_{s}$ and $\bxi_{s}^\top\bH_{i}\bxi_{s}=1$, so each $\bxi_{s}$ is feasible. Since eigenvectors are determined only up to sign, both $\pm\bxi_{s}$ are evaluated. The working update is then selected as
\[
  \hat{\tilde{\bgamma}}_{i}=\underset{\bxi\in\{\pm\bxi_{1},\dots,\pm\bxi_{p}\}}{\arg\min} \bxi^\top\bA_{i}\bxi-\kappa\bgamma^\top\bxi.
\]
The direction used in the likelihood and in the von Mises-Fisher update is then $\hat{\bgamma}_{i}=\hat{\tilde{\bgamma}}_{i}/\|\hat{\tilde{\bgamma}}_{i}\|_2$. For $\kappa=0$ this recovers the exact generalized eigenvalue solution up to Euclidean rescaling. For $\kappa>0$ this eigenvector search gives a finite candidate approximation embedded within the block coordinate descent iterations of Algorithm~\ref{alg:mlcap}, so that the remaining parameters adjust to accommodate the discrete selection at each step.

\section{Additional simulation results}
\label{appendix:sec:sim}

\subsection{Inference on regression coefficients}
\label{appendix:sub:sim_inference}

We evaluate the inference performance following both the proposed bootstrap procedure in Section~\ref{sub:inference} and the asymptotic normality result in Theorem~\ref{thm:asym_known}. Following the same simulation setting as in Section~\ref{sec:sim}, Table~\ref{appendix:table:sim_inference} presents the coverage probability of 95\% confidence intervals. The bootstrap procedure achieves better coverage than the asymptotic normality-based method. As the sample size increases, the performance of the two procedures improves and converges. Neither approach accounts for variation in estimating $\bgamma$'s when performing inference on the model coefficients.

\begin{table}
  \caption{\label{appendix:table:sim_inference}Coverage probability ($\%$) of 95\% confidence intervals for $\beta_{11}^{*}$ and $\beta_{21}^{*}$ based on the bootstrap procedure with $500$ replications (Bootstrap) and the asymptotic normality result (Asymptotic). Data dimension $p=5,20$, number of clusters $m=20$, sample sizes $n=100,500$ and $T=100,500$.}
  \begin{center}
    \begin{tabular}{c c c r r c r r}
      \hline
      & & & \multicolumn{2}{c}{$\hat{\beta}_{11}$} & & \multicolumn{2}{c}{$\hat{\beta}_{21}$} \\
      \cline{4-5} \cline{7-8}
      \multirow{-2}{*}{$p$} & \multirow{-2}{*}{$n$} & \multirow{-2}{*}{$T$} & \multicolumn{1}{c}{Bootstrap} & \multicolumn{1}{c}{Asymptotic} & & \multicolumn{1}{c}{Bootstrap} & \multicolumn{1}{c}{Asymptotic} \\
      \hline
      & $100$ & $100$ & $78.6$ & $82.3$ && $54.7$ & $64.6$ \\
      \multirow{-2}{*}{$5$} & $500$ & $500$ & $94.0$ & $75.0$ && $100.0$ & $100.0$ \\
      \hline
      & $100$ & $100$ & $57.6$ & $40.9$ && $100.0$ & $100.0$ \\
      \multirow{-2}{*}{$20$} & $500$ & $500$ & $97.5$ & $77.0$ && $93.5$ & $95.0$ \\
      \hline
    \end{tabular}
  \end{center} 
\end{table}

\section{Additional results of the HCP Lifespan study}
\label{appendix:sec:hcp}

\subsection{Data}
\label{appendix:sub:hcp_data}

Sample size within age-defined clusters in the HCP Lifespan Study is presented in Table~\ref{appendix:table:hcp_sample}. The total number of subjects is $1,500$.
\begin{table}
  \caption{\label{appendix:table:hcp_sample}Sample size within age-defined clusters in the HCP Lifespan Study.}
  \begin{center}
    \begin{tabular}{c r c c c}
      \hline
      \multicolumn{1}{c}{Age} & \multicolumn{1}{c}{$n_{i}$} & \multicolumn{1}{c}{Female} & \multicolumn{1}{c}{Male} & \multicolumn{1}{c}{Study} \\
      \hline
      $[5,10]$ & $76$ & $49$ & $27$ & HCP-D \\
      $(10,15]$ & $196$ & $103$ & $93$ & HCP-D \\
      $(15,22]$ & $231$ & $121$ & $110$ & HCP-D \\
      $[22,25]$ & $108$ & $41$ & $67$ & HCP-YA \\
      $(25,30]$ & $185$ & $106$ & $79$ & HCP-YA \\
      $(30,35]$ & $141$ & $92$ & $49$ & HCP-YA \\
      $[36,40]$ & $63$ & $30$ & $33$ & HCP-A \\
      $(40,45]$ & $66$ & $40$ & $26$ & HCP-A \\
      $(45,50]$ & $65$ & $34$ & $31$ & HCP-A \\
      $(50,55]$ & $66$ & $39$ & $27$ & HCP-A \\
      $(55,60]$ & $63$ & $35$ & $28$ & HCP-A \\
      $(60,65]$ & $53$ & $27$ & $26$ & HCP-A \\
      $(65,70]$ & $51$ & $29$ & $22$ & HCP-A \\
      $(70,75]$ & $46$ & $25$ & $21$ & HCP-A \\
      $(75,80]$ & $29$ & $12$ & $17$ & HCP-A \\
      $(80,85]$ & $37$ & $20$ & $17$ & HCP-A \\
      $(85,90]$ & $24$ & $12$ & $12$ & HCP-A \\
      \hline
      Total & $1,500$ & $815$ & $685$ & \\
      \hline
    \end{tabular}
  \end{center}
\end{table}

\subsection{Additional results}
\label{appendix:sub:hcp_results}

Figure~\ref{appendix:fig:hcp_gamma_sparse} presents the loading profile of the identified average brain network and Figure~\ref{appendix:fig:hcp_riverplot} presents the river plot of each functional model.

\begin{figure}
  \subfloat[\label{appendix:fig:hcp_gamma_sparse}]{\includegraphics[width=0.66\textwidth]{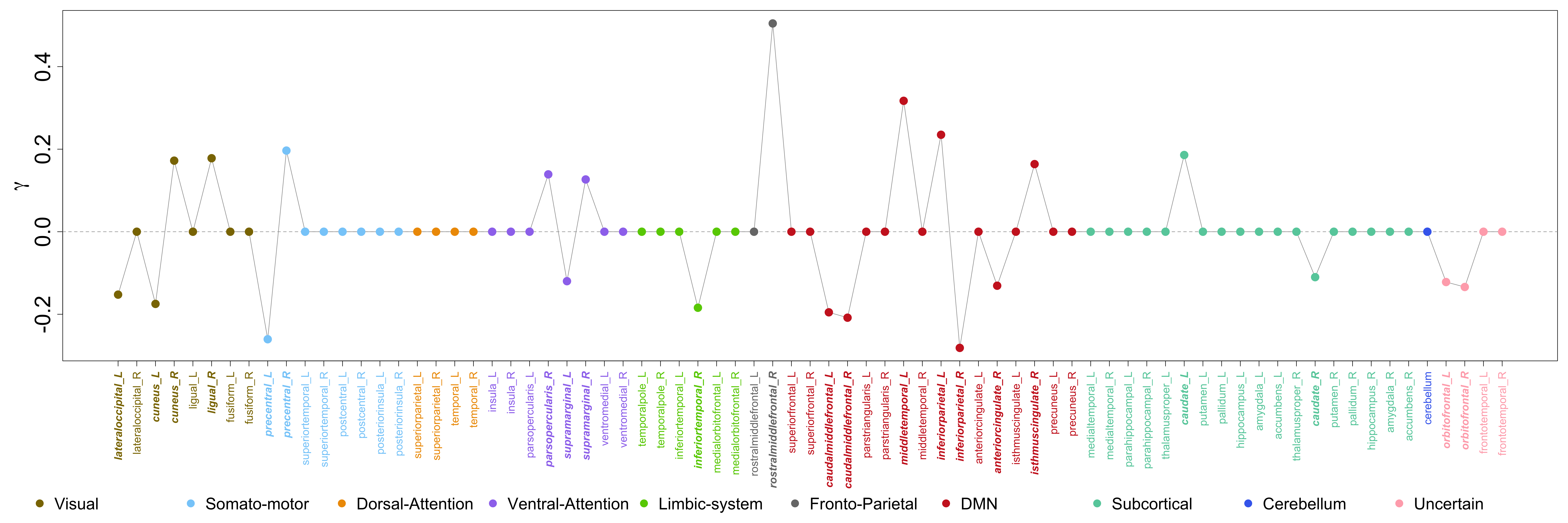}}
  \subfloat[\label{appendix:fig:hcp_riverplot}]{\includegraphics[width=0.33\textwidth]{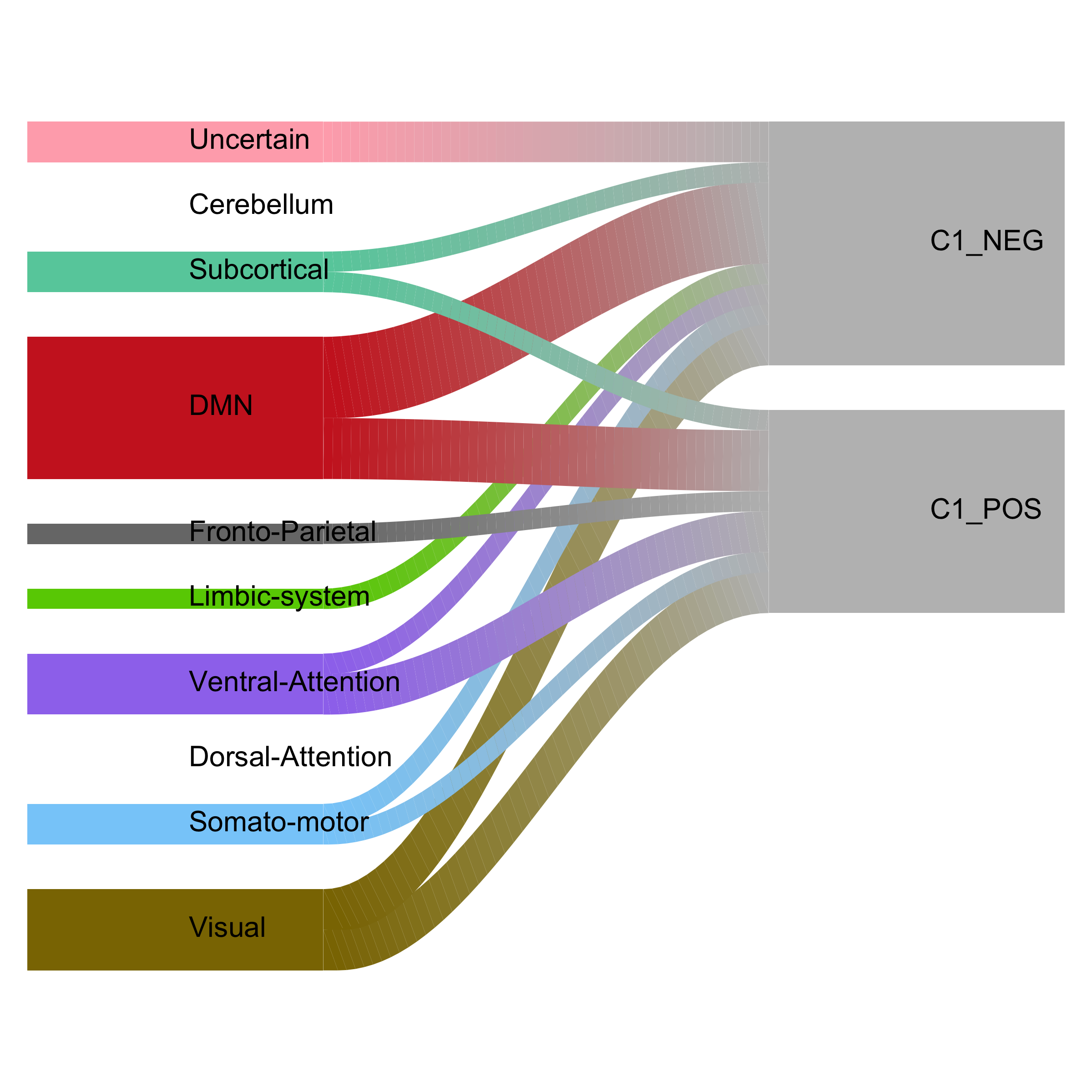}}
  \caption{\label{appendix:fig:hcp_avgbrain}The average projection of the component identified in the HCP study. (a) The river plot of the average projection by functional module. (b) The estimated $\bgamma$ value of the average projection.}
\end{figure}


\bibliographystyle{apalike}
\bibliography{Bibliography}

\end{document}